\documentclass[letterpaper,12pt]{article}
\usepackage[T1]{fontenc}
\usepackage{tabularx, amsmath, amsthm, amsfonts, mathtools, amssymb, physics, graphicx, cite, algorithm, multirow, braket, adjustbox, array} 
\usepackage[margin=1in,letterpaper]{geometry} 
\usepackage[final]{hyperref} 
\hypersetup{
	colorlinks=true,       
	linkcolor=blue,        
	citecolor=blue,        
	filecolor=magenta,     
	urlcolor=blue         
}

\newtheorem{theorem}{Theorem}[section]

\newtheorem{lemma}[theorem]{Lemma}
\newtheorem*{remark}{Remark}

\newtheorem*{notation}{Notation}

\DeclareMathOperator{\inv}{inv}
\DeclareMathOperator{\erase}{erase}
\newcommand*{\defeq}{\stackrel{\text{def}}{=}}
\DeclarePairedDelimiter\floor{\lfloor}{\rfloor}
\DeclarePairedDelimiter\tfloor{\lfloor}{\rfloor _ {T}}
\DeclarePairedDelimiter\tremain{ \{ }{ \} _ {T}}
\DeclarePairedDelimiter\vtremain{ \{ }{ \} _ {VT}}

\def\brcurs{{\mbox{$\resizebox{.09in}{.08in}{\includegraphics[trim= 1em 0 14em 0,clip]{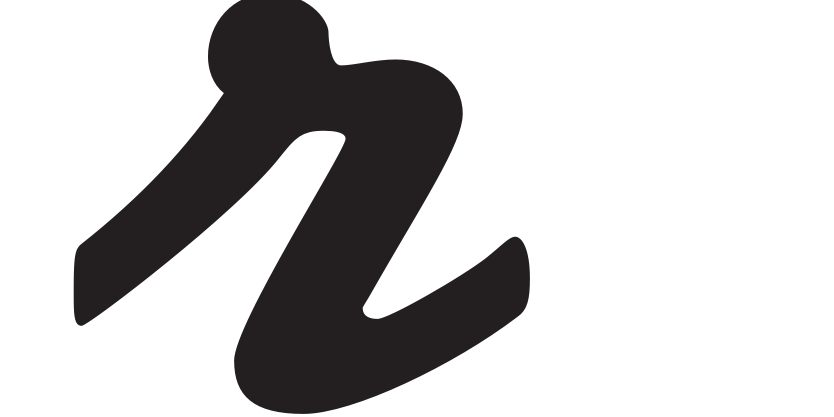}}$}}}

\input ulem.sty\relax
\catcode`\@=11 %

\font\uwavefont=lasyb10 scaled 652

\def\uwave{%
  \bgroup
    \markoverwith{%
      \lower3.5\p@\hbox{\uwavefont\char58}%
    }%
  \ULon
}

\begin{document}

\title{Stop using Landau gauge for Tight-binding Models}
\author{Seungwon Jung\footnote{seungwonjung1@gmail.com}}
\date{}
\maketitle

\begin{abstract}
To analyze the electronic band structure of a two-dimensional (2D) crystal under a commensurate perpendicular magnetic field, tight-binding (TB) Hamiltonians are typically constructed using a magnetic unit cell (MUC), which is composed of several unit cells (UC) to satisfy flux quantization. However, when the vector potential is constrained to the Landau gauge, an additional constraint is imposed on the hopping trajectories, further enlarging the TB Hamiltonian and preventing incommensurate atomic rearrangements. In this paper, we demonstrate that this constraint persists, albeit in a weaker form, for any linear vector potential ($\vb{A}(\vb{r})$ linear in $\vb{r}$). This restriction can only be fully lifted by using a nonlinear vector potential. With a general nonlinear vector potential, a TB Hamiltonian can be constructed that matches the minimal size dictated by flux quantization, even when incommensurate atomic rearrangements occur within the MUC, such as moiré reconstructions. For example, as the twist angle \(\theta\) of twisted bilayer graphene (TBG) approaches zero, the size of the TB Hamiltonian scales as \(1/\theta^4\) when using linear vector potentials (including the Landau gauge). In contrast, with a nonlinear vector potential, the size scales more favorably, as \(1/\theta^2\), making small-angle TBG models more tractable with TB.
\end{abstract}

\section{Introduction}
TB models are widely used on numerical calculation of the electronic band structure because it is simple but expressive. To apply TB to 2D crystals under a uniform perpendicular magnetic field, we use the method introduced by D. Hofstadter\cite{hofstadter}: We enlarge the UC to an MUC such that the flux quantization condition is satisfied, allowing the internal degrees of freedom of the MUC to be wrapped into a torus. While D. Hofstadter's square lattice can be successfully solved using the simple Landau gauge, honeycomb lattice requires additional techniques to reduce the TB Hamiltonian to the exact size determined by the flux quantization\cite{rammal}. However, most of current researchers and software packages still use the Landau-like gauge due to the lack of a general technique to tackle this problem for arbitrary 2D crystals\cite{Kwant, KITE, WannierTools, HofstadterTools, TBPLaS, BerryEasy,GPUQT,TBTK,TopoTB}. The computational complexity of solving the eigenproblem scales from quadratic to cubic with the size of the TB Hamiltonian (depending on the sparsity of the Hamiltonian or the type of eigenproblem). Hence, the use of Landau-like gauge greatly limits the potential of the TB method and creates the misconception that TB is only suitable for simple textbook examples.

Here, we demonstrate that a nonlinear vector potential (Equation \ref{eq:Nonlinear vector potential}) provides a general solution to this problem, yielding a formula (Equation \ref{Eq:Hk for nonlinear vector potential}) that is as simple as the Landau gauge. First, we precisely define the problem in a general context. When constrained to linear vector potentials like the Landau gauge, additional commensurability constraints on hopping trajectories are required: (1) the hopping trajectories must be commensurate (i.e., rank=2), and (2) the size of the TB Hamiltonian becomes significantly larger. We show that the lower bound for the TB Hamiltonian size with general linear vector potentials does not match the exact limit imposed by flux quantization. We also introduce a general form of nonlinear vector potential that only requires flux quantization, significantly expanding the capabilities of the TB method by (1) enabling both commensurate and incommensurate atomic rearrangements without changing MUC, and (2) this makes the size of the MUC to be only bigger than that of UC just by a constant factor, making intractable problems (such as magic-angle moiré-TBG whose twisted angle is as small as $\sim 1.1^{\circ}$) tractable. 

\begin{remark}
Mathematical notations are given in the Chapter \ref{Appendix: Mathematical Notations} (Appendix A). To avoid any misunderstandings, please review this section before reading the main text.
\end{remark}

\section{Crystal structure and Tight-binding Model}
2D crystal is constructed as the infinite repetition of identical groups of atoms. A group is called the basis, while the set of points to which the basis is attached is called the lattice. The primitive translation vectors of the lattice are $\vb{v_1}$ and $\vb{v_2}$ and these are combined as a matrix $V \defeq \mqty(\vb{v_1} & \vb{v_2})$ where $|V|>0$. UCs of the lattice are represented with indices $\vb{j}(\vb{r}) \defeq \floor*{\vb{r}}_{V}$. Basis is composed of atoms whose relative position with respect to the lattice point $V\vb{j}$ is $\brcurs_n~ (n=1,2,...,N_{atom})$. Hopping network for TB is also defined on the basis: $(n_m,~n_m',~\vb{\Delta}_m,~t_m) ~(m=1,2,...,N_{hop})$ denotes that electron hops from atom $n_m$ of UC $\vb{j}$ to atom $n_m'$ of UC $\vb{j}+\vb{\Delta}_m$ with hopping energy $t_m$ and also hops in the reverse direction with hopping energy ${t_m}^{*}$, for every $\vb{j}$. It is assumed that the hopping network consists of a single connected component, ensuring that the TB Hamiltonian cannot be trivially decomposed into a tensor product of smaller Hamiltonians. Then, $(1,0)$ and $(0,1)$ are in $span(S_{hop})$ where $S_{hop} \defeq \{ \vb{\Delta}_m+V^{-1} (\brcurs_{n_m'}-\brcurs_{n_m}) | m=1,2,...,N_{hop} \}$.

Uniform magnetic field $\vb{B}=B\hat{z}=\curl \vb{A}$ is applied on the 2D crystal. Magnetic flux on a unit cell is normalized with magnetic flux quantum as $\Phi \defeq \frac{qB|V|}{h}$. Magnetic (super)lattice is defined with primitive vectors $\vb{u_1}$, $\vb{u_2}$, and the transformation matrix $T \defeq V^{-1} \mqty(\vb{u_1} & \vb{u_2})\in \mathbb{Z}_{2\times2}$ where $|T|>0$. MUCs of the magnetic lattice are represented with indices $\vb{i}(\vb{r}) \defeq \tfloor*{\vb{j(\vb{r})}}$. MUC $\vb{i}$ is composed of UCs $\vb{j}\in T\vb{i} + J$ where a set $J$ is defined as $J \defeq \left\{ \vb{j}\in \mathbb{Z}^2 | \tfloor*{\vb{j}}=\vb{0} \right\}$. Atoms in MUC is represented as $(\vb{j},~n)~(\vb{j}\in J,~n=1,2,...,N_{atom})$ and edges of the hopping network are also inherited likewise as $((\vb{j},~n_m),~(\tremain*{\vb{j}+\vb{\Delta}_m},~n_m'),~\tfloor{\vb{j}+\vb{\Delta_m}},~t_m)~(\vb{j}\in J,~m=1,2,...,N_{hop})$.

\begin{figure}[H]
\label{fig:Crystal}
\centering \includegraphics[width=\columnwidth]{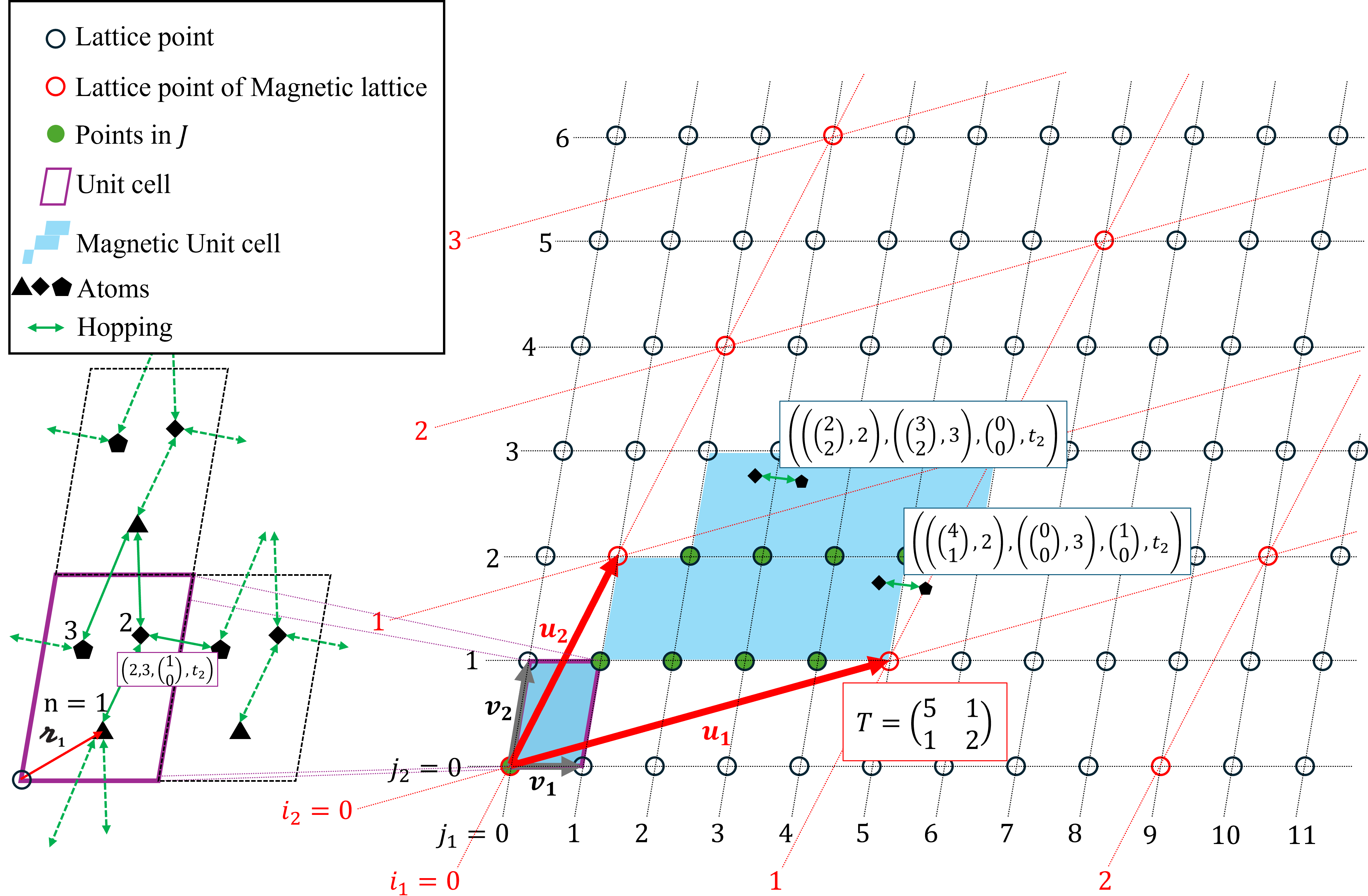}
\caption{
Example of how parameters define the lattice and the magnetic lattice. Crystal in this figure has primitive vectors $\vb{v_1},\vb{v_2}$ and has three atoms ($n=1,2,3$) in its UC (marked with purple boundaries). Indices $j_1$ and $j_2$ are marking its lattice points. Among four hoppings of the crystal ($m=1,2,3,4$), second hopping was described as $\big(2,3,\big( \begin{smallmatrix} 1\\0 \end{smallmatrix}\big),t_2\big)$. Magnetic lattice of the crystal was constructed with $T=\big(\begin{smallmatrix} 5&1\\1&2\end{smallmatrix}\big)$. Magnetic lattice has primitive vectors $\vb{u_1},\vb{u_2}$ and its MUC (shaded in light blue) is composed of 9 UCs, each corresponding to a point in $J$ (marked with green dots). Hopping is inherited from the original lattice and two of them are described in the text boxes with blue boundary lines.
}
\end{figure}

With Peierls substitutition\cite{Peierls}, corresponding TB Hamiltonian is given by
\begin{equation} \label{eq:Real-space Hamiltonian}
    \textrm{H} \defeq -\sum_{\substack{\vb{i}\in\mathbb{Z}^2,~\vb{j}\in J,\\ m=1,2,...,N_{hop}}} t_m
~\exp( i \frac{q}{\hbar} \int_{V(T\vb{i}+\vb{j})+\brcurs_{n_m}}^{V(T\vb{i}+\vb{j}+\vb{\Delta}_m)+\brcurs_{n_m'}} \vb{A} \cdot d\vb{r} )
    ~c_{\vb{i}+\tfloor{\vb{j}+\vb{\Delta_m}},~(\tremain*{\vb{j}+\vb{\Delta}_m},~n_m')}^{\dag} ~c_{\vb{i},~(\vb{j},~n_m)}+h.c.
\end{equation}
where $c_{\vb{i},(\vb{j},n)}$ and $c_{\vb{i},(\vb{j},n)}^{\dag}$ are annihilation and creation operators of the electron at the atom $(\vb{j},n)$ of MUC $\vb{i}$, respectively. The reciprocal lattice of the magnetic lattice is defined in the momentum space with primitive translation vectors composing the columns of $2\pi (V^T)^{-1} (T^T)^{-1}$ and the unit cell, Brillouin zone (BZ). Annihilation and creation operators, $c_{\vb{k},(\vb{j},n)}$ and $c_{\vb{k},(\vb{j},n)}^{\dag}$, are also defined in the BZ and Fourier-transformed into that of real space as
\begin{equation}\begin{split}\label{eq:Fourier transform of operators}
    & c_{\vb{i},(\vb{j},n)} = \dfrac{\sqrt{|V||T|}}{2\pi} \iint_{BZ} d^2 \vb{k}~ c_{\vb{k},(\vb{j},n)} ~\exp(i (VT\vb{i}) \cdot \vb{k}), \\
    & c_{\vb{k},(\vb{j},n)} = \dfrac{\sqrt{|V||T|}}{2\pi} \sum_{\vb{i}} c_{\vb{i},(\vb{j},n)} ~\exp(-i (VT\vb{i}) \cdot \vb{k}).
\end{split}\end{equation}

\section{Linear vector potential}
Assume that the vector potential is constrained to be linear. In other words, there exists a real $2\times 2$ matrix $A$ satisfying $A_{21}-A_{12}=\Phi$ and
\begin{equation}\label{eq:Linear vector potential}
\vb{A} (\vb{r}) = \dfrac{h}{q} (V^T)^{-1} A V^{-1} \vb{r}.
\end{equation}

By the Theorem \ref{theorem:k-space Hamiltonian} (1), if and only if
\begin{equation}\label{cond-linear}
^{\forall} \vb{R} \in S_{hop},~ T^T A^T \vb{R} \in \mathbb{Z}^2,
\end{equation}
$\textrm{H} = \iint_{BZ} d^2\vb{k} H_{\vb{k}}$ where the kernel of the Hamiltonian is
\begin{equation}\begin{split}
        H_{\vb{k}} = -\sum_{\substack{\vb{j}\in J,\\ m=1,2,...,N_{hop}}} t_m
~ &\exp( -i (VT\tfloor*{\vb{j}+\vb{\Delta}_m}) \cdot \vb{k}+2\pi i \int_{\vb{j}+V^{-1} \brcurs_{n_m}}^{\vb{j}+\vb{\Delta}_m+V^{-1}\brcurs_{n_m'}} (A \vb{R}) \cdot d\vb{R} )
    \\ \times & ~c_{\vb{k},~(\tremain*{\vb{j}+\vb{\Delta}_m},~n_m')}^{\dag} ~c_{\vb{k},~(\vb{j},~n_m)}+h.c.~.
\end{split}\end{equation}

To satisfy the Equation \ref{cond-linear}, $rank(S_{hop}) =2$ should be met, because, if $rank(S_{hop}) \geq 3$, we can find an arbitrary small nonzero $\vb{R}\in span(S_{hop})$. Since $rank(S_{hop})=2$, there exists a matrix $D=\mqty(\vb{d}_1 & \vb{d}_2)$ where $span(S_{hop})=span(\{\vb{d}_1,~\vb{d}_2\})$ and
\begin{equation}
    D=\dfrac{1}{\alpha \beta_1 \beta_2 \gamma}\mqty(\beta_1 & \\ \beta_2 \gamma_{21} & \beta_2 \gamma)~ or ~ \dfrac{1}{\alpha \beta_1 \beta_2 \gamma}\mqty(\beta_1 \gamma & \beta_1 \gamma_{12} \\ & \beta_2)
\end{equation}
where positive integers $\alpha,~\beta_1,~\beta_2,~\gamma$ and non-negative integers $\gamma_{12},~\gamma_{21}$ satisfy $(\beta_1,~\beta_2)=1$ and $\gamma_{12} \gamma_{21} \equiv 1 (mod~\gamma)$, as shown in the Table \ref{Table:Parameters for D} for well-known crystals. Now, Equation \ref{cond-linear} becomes
\begin{equation}\label{eq:new cond-linear}
    D^T A T \in \mathbb{Z}_{2\times2},
\end{equation}
which implies $\Phi \in \mathbb{Q}$.

By Theorem \ref{theorem:main theorem for linear vector potential}, the size of the TB Hamiltonian $N_{atom} |T|$ is always a multiple of $N_{atom} q_{\Phi / \alpha}$. A solution $(A,~T)$ satisfying the lower bound $|T|=q_{\Phi / \alpha}$ is given by
\begin{equation}\begin{split}\label{eq: Linear A,T}
A & =\Phi \mqty(\delta \\ \beta_1) \mqty(\inv_{\delta} \beta_1 & -\inv_{\beta_1} \delta)
\\
and~T & =\mqty(\beta_1 & \inv_{\beta_1} \delta \\ -\delta & \inv_{\delta} \beta_1) \mqty(q_{\Phi / \alpha} & \\ & 1),
\\
where~ \delta & \defeq \beta_2 (\gamma (1+ \gamma_{21} \inv_{\erase_{\gamma} \beta_1} \gamma) -\gamma_{21})~and~\beta_1~are~coprime.
\end{split}\end{equation}

With linear vector potentials, the lower bound of the size of the TB Hamiltonian is $N_{atom} q_{\Phi / \alpha}$, which doesn't coincide with the exact lower bound $N_{atom} q_{\Phi}$ given by the flux quantization.

\begin{table}[!h]
\begin{center}
\caption{Parameters defining the matrix $D$}
\label{Table:Parameters for D}
\begin{tabular}{ |p{2.5cm}||p{2cm}|p{0.5cm}|p{.5cm}|p{.5cm}|p{.5cm}|p{.5cm}|   }
 \hline
 Crystal& $\alpha$ & $\beta_1$ & $\beta_2$ & $\gamma$ & $\gamma_{21}$ & $\gamma_{12}$\\
 \hline
 Square & 1&1&1&1&0&0
 \\
 Triangular & 1&1&1&1&0&0
 \\
 Honeycomb&1&1&1&3&1&1
 \\
 Kagome&2&1&1&1&0&0
 \\
 moiré-TBG \ref{moire-TBG} & $(3t^2+1)/2$ & 1 & 1& 3&1&1
 \\
 \hline
\end{tabular}
\end{center}
\end{table}

\subsection{Landau-like vector potential}
Assume that the linear vector potential is further constrained to be Landau-like:
\begin{equation}\label{eq: Landau-like A}
A=\mqty(0&0\\ \Phi & 0)
\end{equation}

Since $D^T A = \dfrac{1}{\alpha \beta_1 \beta_2 \gamma}\mqty(\beta_1 \gamma &\\ \beta_1 \gamma_{12} & \beta_2) \mqty(0&0\\ \Phi & 0)=\dfrac{\Phi}{\alpha \beta_1 \gamma} \mqty(0&0\\1&0)$,
to satisfy the Equation \ref{eq:new cond-linear},
\begin{equation}\label{eq: Landau-like T}
T=\mqty(q_{\frac{\Phi }{ \alpha \beta_1 \gamma}} & \\ & 1)
\end{equation}
is proved to be the best solution according the Theorem \ref{theorem:|T|=q^2/(q,ad-bc)} (2).

The lower bound of the size of the TB Hamiltonian with Landau-like vector potential is $N_{atom} q_{\frac{\Phi }{ \alpha \beta_1 \gamma}}$, which is not enough even for a textbook lattice, the Honeycomb lattice. The honeycomb lattice has non-unity $\beta_1 \gamma = 3$, as shown in the Table \ref{Table:Parameters for D}. Hence, Landau-like vector potential needs three times bigger MUC compared to generalized linear vector potential when $3\nmid p_{\Phi}$, as shown in Figure \ref{fig:Hopnetwork} (a,b) for the $\Phi=\frac{1}{2}$ case. Linear vector potential minimizes the MUC with the diagonal term of $A$, such as $1$ in $A=\big(\begin{smallmatrix} 1&0\\  \frac{1}{2}&0\end{smallmatrix}\big)$ of $\Phi=\frac{1}{2}$ case. This diagonal term is mathematically equivalent to the Rammal transformation of the wavefunctions \cite{rammal}.

\begin{figure}[H]
\centering \includegraphics[width=\columnwidth]{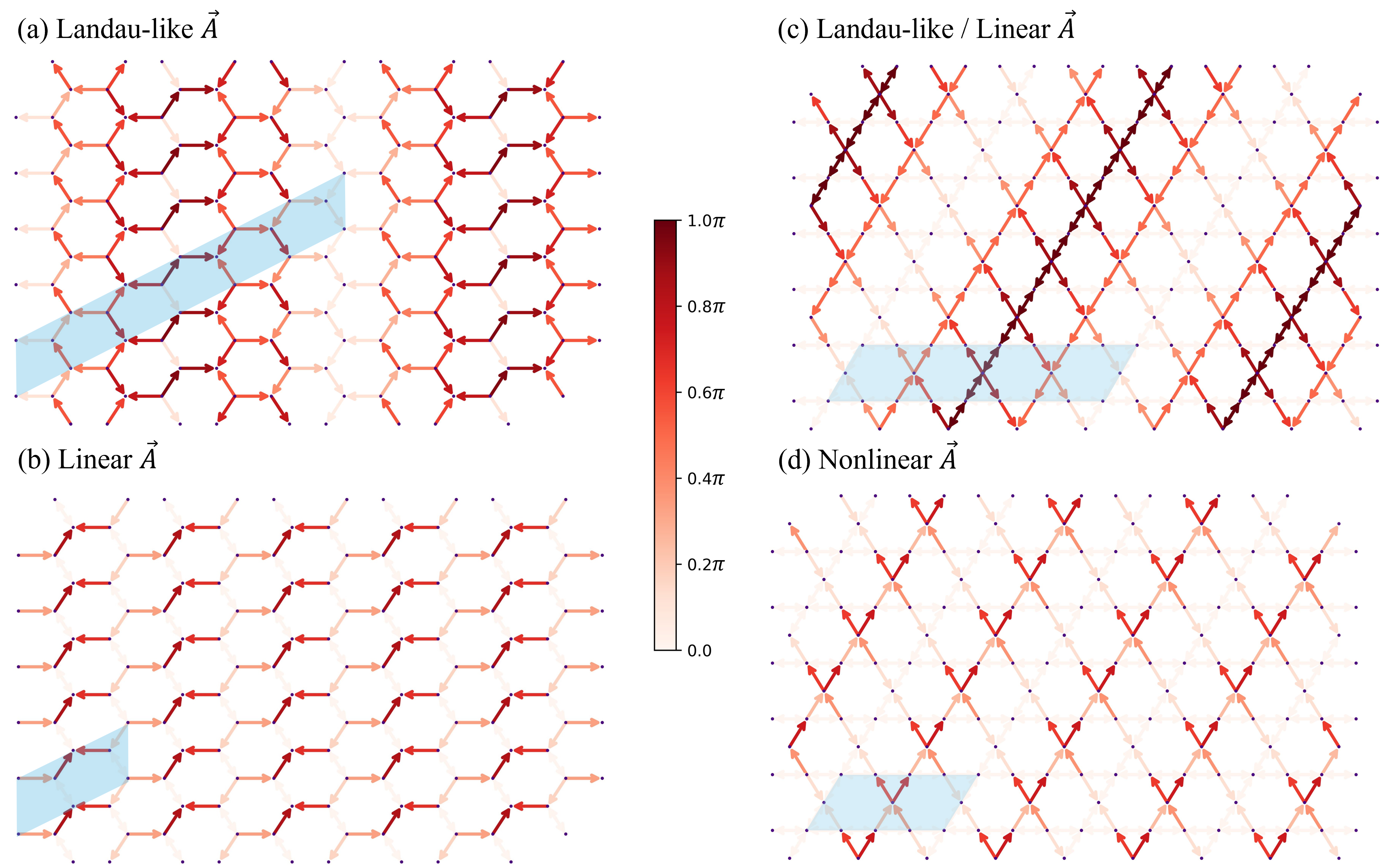}
\caption{
Even textbook-lattices need an application of non-Landau-like linear or nonlinear vector potential to minimize the size of their MUC (shaded in blue). For Honeycomb lattice with $\Phi=\frac{1}{2}$, (a) Landau-like vector potential with $A= \big(\begin{smallmatrix} 0&0\\ \frac{1}{2} &0\end{smallmatrix}\big),~T=\big(\begin{smallmatrix} 6&\\  &1\end{smallmatrix}\big)$ needed MUC three times bigger than (b) linear vector potential with $A= \big(\begin{smallmatrix} 1&0\\ \frac{1}{2} &0\end{smallmatrix}\big),~T=\big(\begin{smallmatrix} 2&\\  &1\end{smallmatrix}\big)$. For Kagome lattice with $\Phi=\frac{1}{2}$, (c) Landau-like/linear vector potential with $A= \big(\begin{smallmatrix} 0&0\\ \frac{1}{2} &0\end{smallmatrix}\big),~T=\big(\begin{smallmatrix} 4&\\  &1\end{smallmatrix}\big)$ needed MUC two times bigger than (d) nonlinear vector potential with $A= \big(\begin{smallmatrix} 0&0\\ \frac{1}{2} &0\end{smallmatrix}\big),~T=\big(\begin{smallmatrix} 2&\\  &1\end{smallmatrix}\big)$. Each figure shows the hopping network whose hopping phase (assumed $\vb{k}=0$) of the edges are depicted with a color scale and the direction of the arrow. (If the arrow is unidirectional with phase $\phi \in (0,\pi)$ depicted as the color of the arrow, the hopping in the reverse direction has the phase $2\pi - \phi$.)
}
\label{fig:Hopnetwork}
\end{figure}

\newpage

\section{Nonlinear vector potential}
The exact lower bound of the size of the TB Hamiltonian given by the flux quantization is $N_{atom} q_{\Phi}$. Linear vector potentials couldn't achieve this exact lower bound. Would we be able to achieve this exact lower bound when we remove all the constraints on the form of the vector potential $\vb{A}$? Here, we present a nonlinear (to be exact, 'can be' nonlinear) vector potential with additional gauge-matching term that can achieve this exact lower bound:
\begin{equation}\begin{split}\label{eq:Nonlinear vector potential}
& \vb{A} (\vb{r}) = \dfrac{h}{q} (V^T)^{-1} A V^{-1} \vb{r}+\nabla \Lambda (\vb{r}),\\
where~ &\Lambda (\vb{r})=-\dfrac{h}{q} ((V^T)^{-1} A T \vb{i}(\vb{r}))\cdot \vtremain*{\vb{r}}
-\dfrac{h}{2q}(\vtremain*{\vb{r}})^T (V^T)^{-1} A V^{-1}\vtremain*{\vb{r}} ,
\end{split}\end{equation}
and $A$ is a real $2\times 2$ matrix satisfying $A_{21}-A_{12}=\Phi$.

By the Theorem \ref{theorem:k-space Hamiltonian} (2), if and only if
\begin{equation}\label{eq:Condition for nonlinear}
T^T A T \in \mathbb{Z}_{2\times 2},
\end{equation}
$\textrm{H} = \iint_{BZ} d^2\vb{k} H_{\vb{k}}$ where the kernel of the Hamiltonian is
\begin{equation}\begin{split}
        H_{\vb{k}} = &-\sum_{\substack{\vb{j}\in J,\\ m=1,2,...,N_{hop}}}  t_m
~ c_{\vb{k},(\tremain*{\vb{j}+\vb{\Delta}_m},~n_m')}^{\dagger} c_{\vb{k},(\vb{j},n_m)}
    \\ &\times \exp \Bigg( 
    -i \vb{k} \cdot (VT\tfloor*{\vb{j}+\vb{\Delta}_m})
    \\ & \indent \indent
    +\pi (\tfloor*{\vb{j}+\vb{\Delta}_m})^T T^T A T (\tfloor*{\vb{j}+\vb{\Delta}_m})
    \\ & \indent \indent
    +\pi \Phi (\tremain*{\vb{j}+\vb{\Delta}_m}+V^{-1}\brcurs_{n_m'})^T \mqty(&-1\\1&) (\vb{j}+V^{-1}\brcurs_{n_m})
    \\ & \indent \indent
    +\pi \Phi (\tfloor*{\vb{j}+\vb{\Delta}_m})^T T^T \mqty(&-1\\1&) (\tremain*{\vb{j}+\vb{\Delta}_m}+V^{-1}\brcurs_{n_m'}+\vb{j}+V^{-1}\brcurs_{n_m})\Bigg) + h.c.~.
\end{split}\end{equation}

Note the symmetry of the terms: when we reverse the hopping direction by exchanging $(\tremain*{\vb{j}+\vb{\Delta}_m}+V^{-1}\brcurs_{n_m'})~\leftrightarrow~(\vb{j}+V^{-1}\brcurs_{n_m})$ and $\tfloor*{\vb{j}+\vb{\Delta}_m}~\leftrightarrow~-\tfloor*{\vb{j}+\vb{\Delta}_m}$, it becomes the phase of the Hermitian conjugate (h.c.) term.

Due to the Lemma \ref{lemma:MMT}, the equation \ref{eq:Condition for nonlinear} implies the flux quantization $\Phi |T| \in \mathbb{Z}$.

\newpage

With a simple Landau-like choice of $A$ and $T$ as
\begin{equation}\label{eq: Nonlinear A,T}
A=\mqty(0&0\\ \Phi & 0),~ T=\mqty(q_{\Phi} & \\ & 1),
\end{equation}
the Equation \ref{eq:Condition for nonlinear} is satisfied and the size of the TB Hamiltonian $N_{atom}|T|=N_{atom} q_{\Phi}$ reaches the exact lower bound given by the flux quantization. In this case, the kernel of the Hamiltonian is simplified as
\begin{equation}\begin{split}\label{Eq:Hk for nonlinear vector potential}
        H_{\vb{k}} = &-\sum_{\substack{\vb{j}\in J,\\ m=1,2,...,N_{hop}}} t_m
~ c_{\vb{k},(\tremain*{\vb{j}+\vb{\Delta}_m},~n_m')}^{\dagger} c_{\vb{k},(\vb{j},n_m)}
    \\ &\times \exp \Bigg( 
    -i \vb{k} \cdot (VT\tfloor*{\vb{j}+\vb{\Delta}_m}) 
    \\ & \indent \indent
    +\pi p_{\Phi} (\tfloor*{\vb{j}+\vb{\Delta}_m})^T \mqty(0&0\\1&0) (\tfloor*{\vb{j}+\vb{\Delta}_m})
    \\ & \indent \indent
    +\pi \Phi (\tremain*{\vb{j}+\vb{\Delta}_m}+V^{-1}\brcurs_{n_m'})^T \mqty(&-1\\1&) (\vb{j}+V^{-1}\brcurs_{n_m})
    \\ & \indent \indent
    +\pi \Phi (\tfloor*{\vb{j}+\vb{\Delta}_m})^T T^T \mqty(&-1\\1&) (\tremain*{\vb{j}+\vb{\Delta}_m}+V^{-1}\brcurs_{n_m'}+\vb{j}+V^{-1}\brcurs_{n_m})\Bigg)+h.c.~.
\end{split}\end{equation}

Since Kagome lattice has non-unity $\alpha=2$, as shown in the Table \ref{Table:Parameters for D}, when $2 \nmid p_{\Phi}$, application of nonlinear vector potential allows to get a MUC with half size of that of the case with Landau-like / linear vector potential, as shown in the Figure \ref{fig:Hopnetwork} (c,d). Figure \ref{fig:Hopnetwork} (c) shows that , with Landau-like / linear vector potential, left half of the MUC looks alike with the right half the MUC. This characteristic is quite general: other crystals with non-unity $\alpha$ also exhibit an MUC with linear vector potentials displaying $\alpha$ similar regions with $2\pi/\alpha$ phase differences. This redundancy is eliminated in the case of nonlinear vector potentials by using a proper gauge transformation with $\Lambda(\vb{r})$. Thus, the MUC and the TB Hamiltonian can be reduced by a factor of $\alpha$ when using nonlinear vector potentials.

\newpage

\section{Example: Twisted Bilayer Graphene}
TBG whose commensuration periodicity matches the moiré periodicity (moiré-TBG\label{moire-TBG}) can be represented with an odd number $t$ corresponding to the twist angle $\theta$\cite{tbg}:
\begin{equation}\begin{split}
\theta=\cos^{-1} \dfrac{3t^2-1}{3t^2+1},~D=\dfrac{2}{3(3t^2+1)}\mqty(1& \\ 1&3),~N_{atom}=3t^2+1,
\\
V=\mqty(\vb{e}_1 & \vb{e}_2)\mqty(-1-t & -2t \\ 2t & -1+t) = \mqty(\vb{e}_1' & \vb{e}_2') \mqty(1-t & -2t \\ 2t & 1+t),
\end{split}\end{equation}
where $\vb{e}_1$ and $\vb{e}_2 = R_{\pi /3} \vb{e}_1$ are primitive vectors of the first graphene, and $\vb{e}_1' = R_{\theta} \vb{e}_1$ and $\vb{e}_2' = R_{\theta} \vb{e}_2$ are primitive vectors of the second graphene which was rotated anti-clockwise with the twisted angle $\theta$.

Since moiré-TBG has non-unity $\beta_1 \gamma =3$ and big $\alpha=\frac{3t^2+1}{2}$, the size of the TB Hamiltonian differs greatly depending on which form of vector potential we are using. Compared to Landau-like vector potential, linear vector potential can reduce the size with a factor of 3, when $3\nmid p_{\Phi}$. More importantly, as the twisted angle $\theta$ decreases to $0$, the lower bound of the size of the TB Hamiltonian with linear vector potentials increases as $\propto 1/ \theta^4$, while the exact lower bound that can be achieved with nonlinear vector potential increases as $\propto 1/ \theta^2$, as shown in the Table \ref{Table:TBG}. Thus, nonlinear vector potentials make previously intractable numerical simulations feasible. As shown in the Figure \ref{fig:TBG} for $\Phi=\frac{5}{7}$ case, with nonlinear vector potentials, moiré-TBG model near the magic-angle TBG (MATBG) is tractable to retrieve 10 eigenpairs (both eigenvalues and eigenvectors) near the Fermi Energy $E_F$ only with the use of nonlinear vector potentials.

Also, it's a lot easier to apply the moiré reconstruction for nonlinear vector potentials. While linear vector potentials required $rank(S_{hop})=2$ and $D^T A T \in \mathbb{Z}_{2\times2}$ for constructing MUC, nonlinear vector potentials required $T^T A T \in \mathbb{Z}_{2\times2}$ which only depends on $\Phi$. Hence, with linear vector potentials, newly relocated atomic positions should be commensurate with the magnetic field to get a TB Hamiltonian, and, even when we relocate atoms nicely to commensurate positions, we have to enlarge the MUC to a huge amount in order to match the updated $D^T A T \in \mathbb{Z}_{2\times2}$ condition. However, with nonlinear vector potentials, we can use the original MUC even with incommensurate atomic positions with $rank(S_{hop})>2$.

\newpage

\begin{table}[!h]
\caption{Size of the Tight-binding Hamiltonian $N_{atom} |T|$}
\begin{adjustbox}{center}
\label{Table:TBG}
\begin{tabular}{ |>{\centering\arraybackslash}m{4.5cm}||>{\centering\arraybackslash}m{4cm}|>{\centering\arraybackslash}m{4cm}|>{\centering\arraybackslash}m{4cm}|  }
 \hline
 \multirow{2}{*}{Case}& Landau-like $\vb{A}$ &
 Linear $\vb{A}$ &
 Nonlinear $\vb{A}$ \\
 & (Eqns. \ref{eq:Linear vector potential}, \ref{eq: Landau-like A}, \ref{eq: Landau-like T}) &
 (Eqns. \ref{eq:Linear vector potential}, \ref{eq: Linear A,T}) &
 (Eqns. \ref{eq:Nonlinear vector potential}, \ref{eq: Nonlinear A,T})
 \\
 \hline
 $3 \nmid p_{\Phi},~g_1 \defeq (p_{\Phi}, \frac{3t^2+1}{2})$&  $\frac{3(3t^2+1)^2}{2g_1} q_{\Phi} \approx \frac{24}{g_1} \frac{q_\Phi}{\theta^4}$  & $\frac{(3t^2+1)^2}{2g_1} q_{\Phi} \approx \frac{8}{g_1} \frac{q_\Phi}{\theta^4}$   &\multirow{2}{*}{$(3t^2+1) q_{\Phi} \approx 4 \frac{q_\Phi}{\theta^2}$}\\
 \cline{1-3}
 $3 \mid p_{\Phi},~g_2 \defeq (\frac{p_{\Phi}}{3},\frac{3t^2+1}{2})$&   $\frac{(3t^2+1)^2}{2g_2} q_{\Phi} \approx \frac{8}{g_2} \frac{q_\Phi}{\theta^4}$  & $\frac{(3t^2+1)^2}{2g_2} q_{\Phi} \approx \frac{8}{g_2} \frac{q_\Phi}{\theta^4}$   & \\
 \hline
\end{tabular}
\end{adjustbox}
\end{table}

\begin{figure}[H]
\centering \includegraphics[width=\columnwidth]{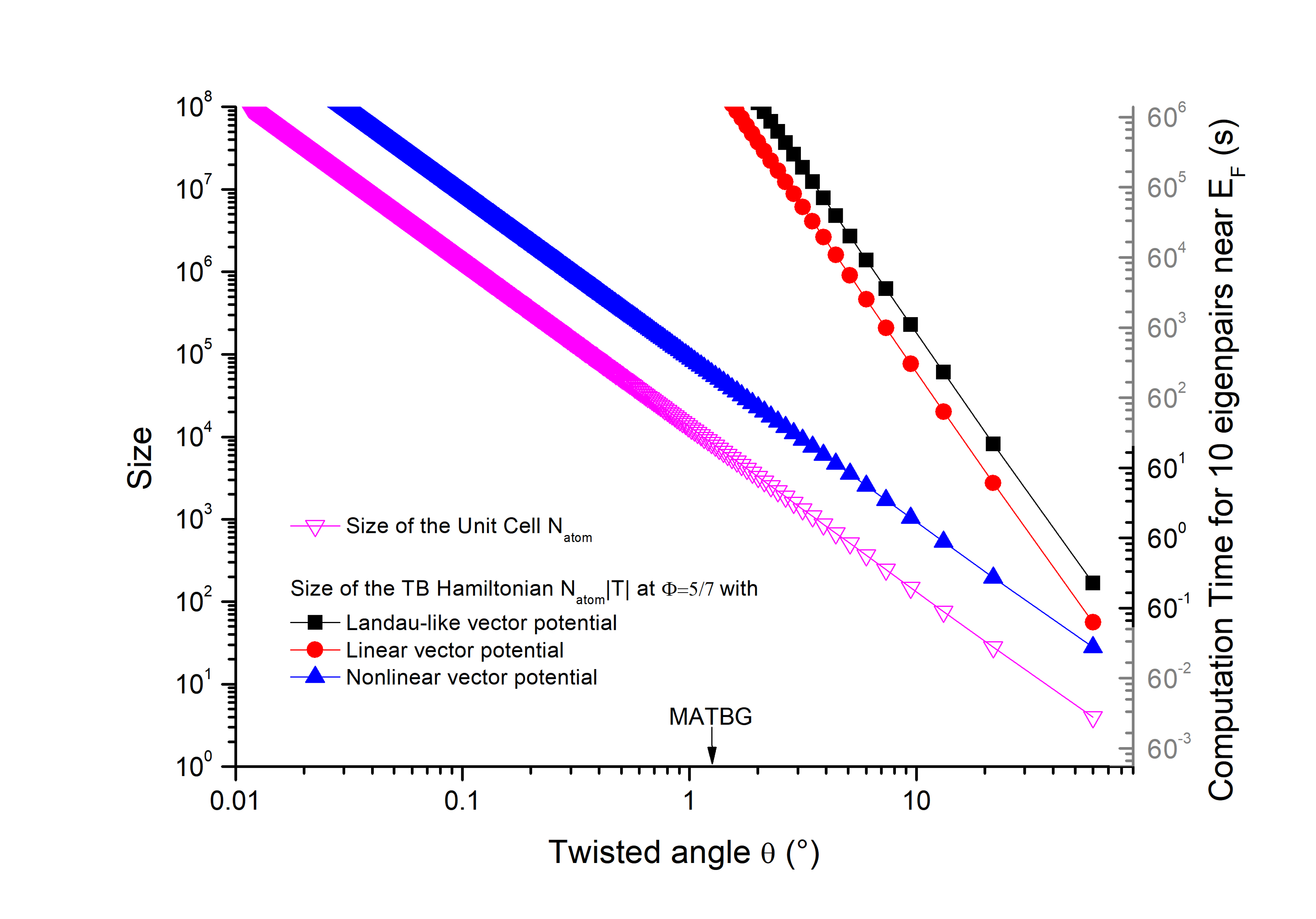}
\caption{
Size of the TB Hamiltonian of the moiré-TBG with uniform magnetic flux $\Phi=\frac{5}{7}$. As the twisted angle $\theta$ decreases to 0, the size of the TB Hamiltonian is $q_{\Phi}=7$ times size of the UC in the case of nonlinear vector potential, while linear vector potentials such as Landau gauge requires quadratically larger TB Hamiltonian. Y-axis on the right shows the computation time for retrieving 10 eigenpairs near the Fermi energy (E=0). Using a laptop with i7-8550U CPU and 16GB RAM, the computation time was measured for the Python scipy.sparse.linalg.eigs function applied to moiré-TBG Hamiltonians, where all hoppings with distances less than $\frac{4}{\sqrt{3}} |\vb{e}_1|$ were considered\cite{Koshino}. The transparent area on the right Y-axis indicates the region where the computation time is estimated through extrapolation. The graph shows that moiré-TBG near MATBG is only tractable with the use of nonlinear vector potentials.
}
\label{fig:TBG}
\end{figure}

\section{Conclusions}
The general formula for the lower bound of the size of the TB Hamiltonian is presented with following constraints on the choice of the vector potential: (1) Landau-like potential, (2) Linear potential, and (3) No constraints. With constraints ((1) and (2)), the hopping trajectories were required to be commensurate and size of the TB Hamiltonian couldn't be always reduced to the exact lower bound given by the flux quantization condition. Linear potentials allowed to have a smaller lower bound compared to that of Landau-like potentials, but it couldn't be able to reach the aforementioned exact lower bound. However, with given general form of nonlinear vector potential with additional gauge-matching term, the size of the TB Hamiltonian could be reduced to this exact lower bound and constraints on hopping trajectories were also removed. While difference between (1) and (2), (2) and (3) were typically shown with honeycomb and kagome lattices, the case of moiré-TBG have shown that the application of nonlinear vector potential can reduce the size of TB Hamiltonian to the square root of that with linear vector potentials, which makes a huge difference in the tractability of band calculations. This work enhances the application of the TB method in two ways: 1) it allows both commensurate and incommensurate atomic reconstructions without altering the MUC, and 2) it ensures that the size of the MUC is only larger by a constant factor than that of the UC by a constant factor defined by the magnetic flux, making previously intractable problems tractable.

\begin{remark}
Check out our \href{https://github.com/expo1221/NonLandauTightBinding}{Github repository} for an interactive explanation. Discover the fastest tight-binding method with nonlinear vector potentials!
\end{remark}

\section*{Acknowledgement}
We would like to express our heartfelt appreciation to Jaejun Yu and Gye-Seon Lee for their invaluable support and guidance throughout this research. We also thank Seokhyun Song and Taeho Kim for their collaboration and insightful feedback.

	\bibliographystyle{ieeetr}
        \bibliography{ref}

\newpage
\section{Appendix A: Mathematical notations}\label{Appendix: Mathematical Notations}
\begin{notation}
Given a finite set of vectors in $\mathbb{R}^2$, $S=\{\vb{v_1},\vb{v_2},...,\vb{v_n}\}$,

$span(S)\defeq\{\sum_{i=1}^{n} c_i \vb{v_i}~|~^{\forall} i=1,2,...,n,~c_i\in \mathbb{Z}\}.$

Note that $span(S)$ is a vector space with $\mathbb{Z}$ as its scalar field. $\mathbb{R}^2$ is a infinite-dimensional vector space with the scalar field $\mathbb{Z}$ and $span(S)$ is its subspace.

$rank(S)\defeq$ (the dimension of $span(S)$) $\leq n$
\end{notation}

\begin{notation}
    Given $r\in \mathbb{R}$, $\floor*{r}\defeq$ the integral part of $r$ and $\{r\}\defeq$ the fractional part of $r$ $=r-\floor*{r}$.
    
    Given $\vb{r}=\mqty(r_1 \\ r_2)\in \mathbb{R}^2$, $\floor*{\vb{r}} \defeq \mqty(\lfloor r_1 \rfloor \\ \lfloor r_2 \rfloor) \in \mathbb{Z}^2~and~\{\vb{r}\}\defeq\mqty(\{r_1\} \\ \{r_2\})$.

    Furthermore, for an invertible $2\times 2$ matrix $L$, $\lfloor \vb{r} \rfloor _{L} \defeq \lfloor L^{-1} \vb{r} \rfloor$ and $\left\{ \vb{r} \right\}_{L} \defeq \vb{r} - L \lfloor \vb{r} \rfloor _{L}$.

    Given $a,b\in\mathbb{Z},~n\in\mathbb{N}$, $a\equiv b~(mod~n)$ denotes that $a$ and $b$ have the same remainder when they are divided by $n$, i.e. $\{ \frac{a-b}{n} \} = 0$ and

    $n|a$ denotes that $a$ is a multiple of $n$, i.e. $\{ \frac{a}{n} \} = 0$.
\end{notation}

\begin{notation}
    Given a matrix $M$ and a value $v\in \mathbb{R}$,
    
    $|M| \defeq$ the determinant of the matrix $M$ and

    $abs(v) \defeq$ the absolute value of the value $v$.
\end{notation}

\begin{notation}
    Given $r=\frac{p}{q}\in \mathbb{Q}$ while $q\in \mathbb{N},~p\in\mathbb{Z}$,

    $q_r\defeq$ the denominator $q$ of the rational number $r$ and

    $p_r\defeq$ the numerator $p$ of the rational number $r$.

    If $r=0$, $q_r \defeq 1$ and $p_r \defeq 0$.
\end{notation}

\begin{notation}
    Given $a_1,~a_2,~...,~a_n\in\mathbb{Z}$,

    $(a_1,~a_2,~,...,~a_n)\defeq$ the greatest common divisor (GCD) of these integers.

    Note that the GCD value is always $\geq 0$ while $(0,a)\defeq abs(a)$ and $(0,0) \defeq 0$.

    Given $r_1,~r_2,~...,~r_n\in\mathbb{Q}$,
    
    let $N\in\mathbb{N}$ to be the least common multiple (LCM) of $q_{r_1},...,q_{r_n}$ such that $^{\forall} i,~r_i N \in \mathbb{Z}$. 
    
    Then,
    $(r_1,~r_2,~,...,~r_n)\defeq \frac{(r_1 N,~r_2 N,...,r_n N)}{N}$.
\end{notation}

\begin{notation}
    Given $a,~b\in\mathbb{Q}$,

    $(\inv_b a,~\inv_a b) \defeq$ integer solution $(x,y)$ of the Diophantine's equation $ax+by=(a,b)$

    i.e. $\inv_b a$ and $\inv_a b$ are the Bézout coefficients of $a$ and $b$.

    Since the solution is not unique, the formula $f$ contains some ambiguity using these definitions as:
    
    $f(\inv_b a)$ represents the set $\{f(x+n\frac{b}{(a,b)})~|~n\in\mathbb{Z}\}$ and

    $f(\inv_b a,\inv_a b)$ represents the set $\{f(x+n\frac{b}{(a,b)},y-n\frac{a}{(a,b)})~|~n\in\mathbb{Z})\}$.

    Note that $a \in \inv_b (\inv_b a)$.
\end{notation}

\newpage

\begin{notation}

    Given an $n\times m$ matrix $M$,

    $M_{ij}\defeq$ an (i,j) entry of the matrix M $(1\leq i \leq n, 1\leq j \leq m)$.    
    
    If this matrix is rational, i.e. $M\in \mathbb{Q}_{n\times m}$,

    $M_{i\ast}\defeq (M_{i1},M_{i2},...,M_{im}),~
    M_{\ast j}\defeq (M_{1j},M_{2j},...,M_{nj})$,
    
    $\gcd M \defeq (M_{11},M_{12},...,M_{1m},M_{21},...,M_{nm})$, 
    $q_M\defeq q_{\gcd M}$, $p_M\defeq p_{\gcd M}$, and
    
    $\Tilde{M} \defeq \frac{1}{\gcd M} M \in \mathbb{Z}_{n\times m}$, where $\gcd \Tilde{M}=1$.
\end{notation}

\begin{notation}
    Given $a,~b\in\mathbb{Z}$,

    $\erase_a b \defeq \max \{n \in \mathbb{N} : n|b~and~(n,a)=1\}$.

    For example, when we erase the prime factors $2,3,5$ of $120=2^3 \times 3 \times 5$ from $-1386=-2\times 3^2\times 7\times 11$, $\erase_{120} (-1386)=7\times 11=77$. 
\end{notation}

\newpage

\section{Appendix B: Theorems and Lemmas}

\begin{theorem}\label{theorem:k-space Hamiltonian}
Real-space Hamiltonian is given as the Equation \ref{eq:Real-space Hamiltonian}.
\begin{equation} \nonumber
    \textrm{H} \defeq -\sum_{\substack{\vb{i}\in\mathbb{Z}^2,~\vb{j}\in J,\\ m=1,2,...,N_{hop}}} t_m
~\exp( i \frac{q}{\hbar} \int_{V(T\vb{i}+\vb{j})+\brcurs_{n_m}}^{V(T\vb{i}+\vb{j}+\vb{\Delta}_m)+\brcurs_{n_m'}} \vb{A} \cdot d\vb{r} )
    ~c_{\vb{i}+\tfloor{\vb{j}+\vb{\Delta_m}},~(\tremain*{\vb{j}+\vb{\Delta}_m},~n_m')}^{\dag} ~c_{\vb{i},~(\vb{j},~n_m)}+h.c.
\end{equation}
\bigskip
(1) If the vector potential is given to be linear as the Equation \ref{eq:Linear vector potential},
\begin{equation}\nonumber
\vb{A} (\vb{r}) = \dfrac{h}{q} (V^T)^{-1} A V^{-1} \vb{r},
\end{equation}
if and only if
\begin{equation}\nonumber
^{\forall} \vb{R} \in S_{hop},~ T^T A^T \vb{R} \in \mathbb{Z}^2,
\end{equation}
$\textrm{H} = \iint_{BZ} d^2\vb{k} H_{\vb{k}}$ where the kernel of the Hamiltonian is
\begin{equation}\begin{split}\nonumber
        H_{\vb{k}} = -\sum_{\substack{\vb{j}\in J,\\ m=1,2,...,N_{hop}}} t_m
~ &\exp( -i (VT\tfloor*{\vb{j}+\vb{\Delta}_m}) \cdot \vb{k}+2\pi i \int_{\vb{j}+V^{-1} \brcurs_{n_m}}^{\vb{j}+\vb{\Delta}_m+V^{-1}\brcurs_{n_m'}} (A \vb{R}) \cdot d\vb{R} )
    \\ \times & ~c_{\vb{k},~(\tremain*{\vb{j}+\vb{\Delta}_m},~n_m')}^{\dag} ~c_{\vb{k},~(\vb{j},~n_m)}+h.c.~.
\end{split}\end{equation}
\bigskip
(2) If the vector potential is given as the Equation \ref{eq:Nonlinear vector potential},
\begin{equation}\begin{split}\nonumber
& \vb{A} (\vb{r}) = \dfrac{h}{q} (V^T)^{-1} A V^{-1} \vb{r}+\nabla \Lambda (\vb{r}),\\
where~ &\Lambda (\vb{r})=-\dfrac{h}{q} ((V^T)^{-1} A T \vb{i}(\vb{r}))\cdot \vtremain*{\vb{r}}
-\dfrac{h}{2q}(\vtremain*{\vb{r}})^T (V^T)^{-1} A V^{-1}\vtremain*{\vb{r}} ,
\end{split}\end{equation}
if and only if
\begin{equation}\nonumber
T^T A T \in \mathbb{Z}_{2\times 2},
\end{equation}
$\textrm{H} = \iint_{BZ} d^2\vb{k} H_{\vb{k}}$ where the kernel of the Hamiltonian is
\begin{equation}\begin{split}\nonumber
        H_{\vb{k}} = &-\sum_{\substack{\vb{j}\in J,\\ m=1,2,...,N_{hop}}} t_m
~ c_{\vb{k},(\tremain*{\vb{j}+\vb{\Delta}_m},~n_m')}^{\dagger} c_{\vb{k},(\vb{j},n_m)}
    \\ &\times \exp \Bigg( 
    -i \vb{k} \cdot (VT\tfloor*{\vb{j}+\vb{\Delta}_m})
    \\ & \indent \indent
    +\pi (\tfloor*{\vb{j}+\vb{\Delta}_m})^T T^T \dfrac{A+A^T}{2} T (\tfloor*{\vb{j}+\vb{\Delta}_m})
    \\ & \indent \indent
    +2 \pi (\tremain*{\vb{j}+\vb{\Delta}_m}+V^{-1}\brcurs_{n_m'})^T \dfrac{A-A^T}{2} (\vb{j}+V^{-1}\brcurs_{n_m})
    \\ & \indent \indent
    +2 \pi (\tfloor*{\vb{j}+\vb{\Delta}_m})^T T^T \dfrac{A-A^T}{2} (\tremain*{\vb{j}+\vb{\Delta}_m}+V^{-1}\brcurs_{n_m'}+\vb{j}+V^{-1}\brcurs_{n_m})\Bigg)+h.c.~.
\end{split}\end{equation}
\end{theorem}
\begin{proof}
(1)
Since $\vb{A}(VT\vb{i}+\vb{r})=\frac{h}{q} (V^T)^{-1} A T\vb{i}+\vb{A}(\vb{r}),$
    \begin{equation}\begin{split}\nonumber
        \theta_{lin}^{\vb{r}}\defeq & \frac{q}{\hbar} \int_{V(T\vb{i}+\vb{j})+\brcurs_{n_m}}^{V(T\vb{i}+\vb{j}+\vb{\Delta}_m)+\brcurs_{n_m'}} \vb{A} \cdot d\vb{r}
        \\
        = &\frac{q}{\hbar} \int_{V\vb{j}+\brcurs_{n_m}}^{V(\vb{j}+\vb{\Delta}_m)+\brcurs_{n_m'}} \vb{A} \cdot d\vb{r} + i (VT\vb{i})\cdot ( 2 \pi (V^T)^{-1} A^T (\vb{\Delta}_m + V^{-1} (\brcurs_{n_m'}-\brcurs_{n_m}))).
    \end{split}\end{equation}

According to the Equation \ref{eq:Fourier transform of operators},
    \begin{equation}\begin{split}\nonumber
        \textrm{H} = & -\sum_{\substack{~\vb{j}\in J,\\ m=1,2,...,N_{hop}}} \iint_{BZ} \iint_{BZ} d^2\vb{k}d^2\vb{k}'
    ~t_m c_{\vb{k},(\tremain*{\vb{j}+\vb{\Delta}_m},~n_m')}^{\dagger} c_{\vb{k}',(\vb{j},n_m)}
    \\ &\times \exp( 2\pi i  \int_{V\vb{j}+ \brcurs_{n_m}}^{V(\vb{j}+\vb{\Delta}_m)+\brcurs_{n_m'}} (AV^{-1}\vb{r})\cdot (V^{-1}d\vb{r})
    -i(VT\tfloor*{\vb{j}+\vb{\Delta}_m})\cdot \vb{k})
    \\ &\times
    \dfrac{|V||T|}{4\pi^2} \sum_{\vb{i}\in\mathbb{Z}^2} \exp(i (VT\vb{i}) \cdot (\vb{k}'-\vb{k}+2\pi (V^T)^{-1} A^T (\vb{\Delta}_m+V^{-1} (\brcurs_{n_m'}-\brcurs_{n_m}))))+h.c.
    \\ =& -\sum_{\substack{~\vb{j}\in J,\\ m=1,2,...,N_{hop}}} \iint_{BZ} \iint_{BZ} d^2\vb{k}d^2\vb{k}'
    ~t_m c_{\vb{k},(\tremain*{\vb{j}+\vb{\Delta}_m},~n_m')}^{\dagger} c_{\vb{k}',(\vb{j},n_m)}
    \\&\times\exp( 2\pi i \int_{\vb{j}+V^{-1} \brcurs_{n_m}}^{\vb{j}+\vb{\Delta}_m+V^{-1}\brcurs_{n_m'}} (A\vb{R}) \cdot d\vb{R}
    -i(VT\tfloor*{\vb{j}+\vb{\Delta}_m})\cdot \vb{k})
    \\ &\times\delta(\vb{k}'-\vb{k}+2\pi (V^T)^{-1} A^T (\vb{\Delta}_m+V^{-1} (\brcurs_{n_m'}-\brcurs_{n_m})))+h.c.
    \\=& -\iint_{BZ} d^2\vb{k} \sum_{\substack{\vb{j}\in J,\\ m=1,2,...,N_{hop}}} 
    ~t_m c_{\vb{k},(\tremain*{\vb{j}+\vb{\Delta}_m},~n_m')}^{\dagger} c_{\vb{k}-2\pi (V^T)^{-1} A^T (\vb{\Delta}_m+V^{-1} (\brcurs_{n_m'}-\brcurs_{n_m})),(\vb{j},n_m)}
    \\
    &\times\exp( 2\pi i \int_{\vb{j}+V^{-1} \brcurs_{n_m}}^{\vb{j}+\vb{\Delta}_m+V^{-1}\brcurs_{n_m'}} (A\vb{R}) \cdot d\vb{R}
    -i(VT\tfloor*{\vb{j}+\vb{\Delta}_m})\cdot \vb{k})+h.c.
    \end{split}\end{equation}

    If and only if
    \begin{equation}\begin{split}\nonumber
        & ^{\forall} m,~ 2\pi (V^T)^{-1} A^T (\vb{\Delta}_m+V^{-1} (\brcurs_{n_m'}-\brcurs_{n_m})) \in 2\pi (V^T)^{-1} (T^T)^{-1}\times \mathbb{Z}^2
        \\ \Leftrightarrow &~ ^{\forall} \vb{R} \in S_{hop},~ T^T A^T \vb{R} \in \mathbb{Z}^2
    \end{split}\end{equation}
    holds, since there is no $\vb{k}$-mixing term in the Hamiltonian, the Hamiltonian is simplified as $\textrm{H} = \iint_{BZ} d^2\vb{k} H_{\vb{k}}$ where the kernel of the Hamiltonian $\textrm{H}_{\vb{k}}$ is
    \begin{equation}\begin{split}\nonumber
        H_{\vb{k}} = -\sum_{\substack{\vb{j}\in J,\\ m=1,2,...,N_{hop}}} t_m
~ &\exp( -i (VT\tfloor*{\vb{j}+\vb{\Delta}_m}) \cdot \vb{k}+2\pi i \int_{\vb{j}+V^{-1} \brcurs_{n_m}}^{\vb{j}+\vb{\Delta}_m+V^{-1}\brcurs_{n_m'}} (A \vb{R}) \cdot d\vb{R} )
    \\ \times & ~c_{\vb{k},~(\tremain*{\vb{j}+\vb{\Delta}_m},~n_m')}^{\dag} ~c_{\vb{k},~(\vb{j},~n_m)}+h.c.~.
\end{split}\end{equation}

For further analysis, let's define the term inside the $\exp (\cdot)$ as $i\theta_{lin}^{\vb{k}}$.

\bigskip

(2)
    Let's divide $A$ into symmetric and antisymmetric parts.
    \begin{equation}\nonumber
    A_{sym} \defeq \dfrac{1}{2} (A+A^T),~ A_{asym} \defeq \dfrac{1}{2} (A-A^T), A=A_{sym}+A_{asym}
    \end{equation}

    Also, for simplicity, we define the braket notation as
    \begin{equation}\begin{split}\nonumber
        \braket{X|Y} &\defeq \ket{X} \cdot (A\ket{Y}),
        \\
        \braket{X|Y}_{sym} &\defeq \ket{X} \cdot (A_{sym}\ket{Y}),
        \\
        \braket{X|Y}_{asym} &\defeq \ket{X} \cdot (A_{asym}\ket{Y}),
        \\
        \ket{X+Y} & \defeq \ket{X}+\ket{Y}
        \\
        \ket{I} &\defeq T \vb{i},
        \\
        \ket{R_i} &\defeq \vb{j}+V^{-1} \brcurs_{n_m},
        \\
        \ket{R_f} &\defeq \tremain*{\vb{j}+\vb{\Delta}_m}+V^{-1} \brcurs_{n_m'},
        \\
        \ket{\Delta} &\defeq T \tfloor*{\vb{j}+\vb{\Delta}_m}.
    \end{split}\end{equation}
    
    Then,
    \begin{equation}\begin{split}\nonumber
        \dfrac{\theta_{non}^{\vb{r}}}{2\pi} \defeq & \frac{q}{h} \int_{V(T\vb{i}+\vb{j})+\brcurs_{n_m}}^{V(T\vb{i}+\vb{j}+\vb{\Delta}_m)+\brcurs_{n_m'}} \vb{A} \cdot d\vb{r}
        \\
        =~~& \int_{T\vb{i}+\vb{j}+V^{-1}\brcurs_{n_m}}
        ^{T\vb{i}+\vb{j}+\vb{\Delta}_m+V^{-1}\brcurs_{n_m'}}
        (A\vb{R})\cdot d\vb{R} 
        + \Lambda(V(T\vb{i}+\vb{j}+\vb{\Delta}_m)+\brcurs_{n_m'})
        - \Lambda(V(T\vb{i}+\vb{j})+\brcurs_{n_m})
        \\
        =~~& \braket{\Delta+R_f-R_i | I+R_i} + \dfrac{1}{2} \braket{\Delta+R_f-R_i | \Delta+R_f-R_i}
        \\&
        -\braket{R_f|I+\Delta} - \dfrac{1}{2} \braket{R_f |R_f}
        +\braket{R_i|I} - \dfrac{1}{2} \braket{R_i|R_i}
        \\
        =~~& \braket{\Delta|I}+\dfrac{1}{2} \braket{\Delta | \Delta}_{sym}
        +\braket{R_f|R_i}_{asym}
        +\braket{\Delta | R_f +R_i}_{asym}
        \\
        = ~~& \vb{i}^T T^T A^T T  \tfloor*{\vb{j}+\vb{\Delta}_m}
        \\+&\dfrac{1}{2} (\tfloor*{\vb{j}+\vb{\Delta}_m})^T T^T A_{sym} T (\tfloor*{\vb{j}+\vb{\Delta}_m})
        \\+& (\tremain*{\vb{j}+\vb{\Delta}_m}+V^{-1}\brcurs_{n_m'})^T A_{asym} (\vb{j}+V^{-1}\brcurs_{n_m})
        \\+& (\tfloor*{\vb{j}+\vb{\Delta}_m})^T T^T A_{asym} (\tremain*{\vb{j}+\vb{\Delta}_m}+V^{-1}\brcurs_{n_m'}+\vb{j}+V^{-1}\brcurs_{n_m})
    \end{split}\end{equation}

    Note that only the first term is dependent on $\vb{i}$. Also note that, when we exchange $(\tremain*{\vb{j}+\vb{\Delta}_m}+V^{-1}\brcurs_{n_m'})~\leftrightarrow~(\vb{j}+V^{-1}\brcurs_{n_m})$ and change the sign of $\tfloor*{\vb{j}+\vb{\Delta}_m}$ (i.e. exchange $\ket{R_f}~\leftrightarrow~\ket{R_i}$ and change the sign of $\ket{\Delta}$), third and fourth term changes its sign, while the second term maintains its sign.

According to the Equation \ref{eq:Fourier transform of operators},
    \begin{equation}\begin{split}\nonumber
        \textrm{H} = & -\sum_{\substack{~\vb{j}\in J,\\ m=1,2,...,N_{hop}}} \iint_{BZ} \iint_{BZ} d^2\vb{k}d^2\vb{k}'
    ~t_m c_{\vb{k},(\tremain*{\vb{j}+\vb{\Delta}_m},~n_m')}^{\dagger} c_{\vb{k}',(\vb{j},n_m)}
    \\ &\times \exp \Bigg( 
    -i \vb{k} \cdot (VT\tfloor*{\vb{j}+\vb{\Delta}_m})
    \\ & \indent \indent
    +\pi (\tfloor*{\vb{j}+\vb{\Delta}_m})^T T^T A_{sym} T (\tfloor*{\vb{j}+\vb{\Delta}_m})
    \\ & \indent \indent
    +2 \pi (\tremain*{\vb{j}+\vb{\Delta}_m}+V^{-1}\brcurs_{n_m'})^T A_{asym} (\vb{j}+V^{-1}\brcurs_{n_m})
    \\ & \indent \indent
    +2 \pi (\tfloor*{\vb{j}+\vb{\Delta}_m})^T T^T A_{asym} (\tremain*{\vb{j}+\vb{\Delta}_m}+V^{-1}\brcurs_{n_m'}+\vb{j}+V^{-1}\brcurs_{n_m})\Bigg)
    \\ &\times
    \dfrac{|V||T|}{4\pi^2} \sum_{\vb{i}\in\mathbb{Z}^2} \exp(i (VT\vb{i}) \cdot (\vb{k}'-\vb{k}+2\pi (V^T)^{-1} A^T T \tfloor*{\vb{j}+\vb{\Delta}_m}))+h.c.
    \\ =& -\sum_{\substack{~\vb{j}\in J,\\ m=1,2,...,N_{hop}}} \iint_{BZ} \iint_{BZ} d^2\vb{k}d^2\vb{k}'
    ~t_m c_{\vb{k},(\tremain*{\vb{j}+\vb{\Delta}_m},~n_m')}^{\dagger} c_{\vb{k}',(\vb{j},n_m)}
    \times\exp(~same~as~above~)
    \\ &\times\delta(\vb{k}'-\vb{k}+2\pi (V^T)^{-1} A^T T \tfloor*{\vb{j}+\vb{\Delta}_m})+h.c.
    \\=& -\iint_{BZ} d^2\vb{k} \sum_{\substack{\vb{j}\in J,\\ m=1,2,...,N_{hop}}} 
    ~t_m c_{\vb{k},(\tremain*{\vb{j}+\vb{\Delta}_m},~n_m')}^{\dagger} c_{\vb{k}-2\pi (V^T)^{-1} A^T T \tfloor*{\vb{j}+\vb{\Delta}_m},(\vb{j},n_m)}
    \times\exp(~same~as~above~)+h.c.
    \end{split}\end{equation}

    If and only if
    \begin{equation}\begin{split}\nonumber
        & ^{\forall} m,~ 2\pi (V^T)^{-1} A^T T \tfloor*{\vb{j}+\vb{\Delta}_m} \in 2\pi (V^T)^{-1} (T^T)^{-1}\times \mathbb{Z}^2
        \\ \Leftrightarrow &~ T^T A T \in \mathbb{Z}_{2\times 2}
    \end{split}\end{equation}
    holds, since there is no $\vb{k}$-mixing term in the Hamiltonian, the Hamiltonian is simplified as $\textrm{H} = \iint_{BZ} d^2\vb{k} H_{\vb{k}}$ where the kernel of the Hamiltonian $\textrm{H}_{\vb{k}}$ is
    \begin{equation}\begin{split}\nonumber
        H_{\vb{k}} = &-\sum_{\substack{\vb{j}\in J,\\ m=1,2,...,N_{hop}}}  t_m
~ c_{\vb{k},(\tremain*{\vb{j}+\vb{\Delta}_m},~n_m')}^{\dagger} c_{\vb{k},(\vb{j},n_m)}
    \\ &\times \exp \Bigg( 
    -i \vb{k} \cdot (VT\tfloor*{\vb{j}+\vb{\Delta}_m})
    \\ & \indent \indent
    +\pi (\tfloor*{\vb{j}+\vb{\Delta}_m})^T T^T A_{sym} T (\tfloor*{\vb{j}+\vb{\Delta}_m})
    \\ & \indent \indent
    +2 \pi (\tremain*{\vb{j}+\vb{\Delta}_m}+V^{-1}\brcurs_{n_m'})^T A_{asym} (\vb{j}+V^{-1}\brcurs_{n_m})
    \\ & \indent \indent
    +2 \pi (\tfloor*{\vb{j}+\vb{\Delta}_m})^T T^T A_{asym} (\tremain*{\vb{j}+\vb{\Delta}_m}+V^{-1}\brcurs_{n_m'}+\vb{j}+V^{-1}\brcurs_{n_m})\Bigg)+h.c.~.
\end{split}\end{equation}

For further analysis, let's define the term inside the $\exp (\cdot)$ as $i\theta_{non}^{\vb{k}}$.

\newpage

(1 and 2)

According to the braket notation defined above, the hopping phase in real-space(${ }^{\vb{r}}$) and reciprocal-space(${ }^{\vb{k}}$) TB Hamiltonian of [(1) linear case](${ }_{lin}$) and [(2) nonlinear case](${ }_{non}$) are
\begin{equation}\begin{split}\nonumber
\theta_{lin}^{\vb{r}}=&\underline{2\pi\braket{\Delta+R_f-R_i|I}}+\uwave{\pi \braket{\Delta+R_f-R_i|\Delta}} + \pi \braket{\Delta|R_f+R_i} + 2\pi \braket{R_f|R_i}_{asym}
        \\&+\pi(\braket{R_f|R_f}_{sym}-\braket{R_i|R_i}_{sym}),
        \\
\theta_{lin}^{\vb{k}}=&-\vb{k}\cdot V\ket{I}+\uwave{\pi \braket{\Delta+R_f-R_i|\Delta}} + \pi \braket{\Delta|R_f+R_i} + 2\pi \braket{R_f|R_i}_{asym}
        \\&+\pi(\braket{R_f|R_f}_{sym}-\braket{R_i|R_i}_{sym}),
        \\
        \theta_{non}^{\vb{r}}
        =~~& \underline{2\pi \braket{\Delta|I}}+\uwave{\pi \braket{\Delta | \Delta}_{sym}}
        +2 \pi \braket{\Delta | R_f +R_i}_{asym}
        +2 \pi \braket{R_f|R_i}_{asym}
        ,
        \\
        \theta_{non}^{\vb{k}}
        =~~& -\vb{k} \cdot V\ket{I}+\uwave{\pi \braket{\Delta | \Delta}_{sym}}
        +2\pi \braket{\Delta | R_f +R_i}_{asym}
        +2\pi \braket{R_f|R_i}_{asym}
        .
\end{split}\end{equation}

The $\ket{I}$-dependent $\underline{underlined}$ terms of real-space phases led to jump in $\vb{k}$ under the Fourier-transform. This led to the requirements, $^{\forall} \vb{R} \in S_{hop},~ T^T A^T \vb{R} \in \mathbb{Z}^2$ and $T^T A T \in \mathbb{Z}_{2\times 2}$, respectively, to make these jumps to happen between equivalent $\vb{k}$ points in the BZ.

Here, we present another meaning of these requirements. Let's reverse the hopping direction by exchanging $\ket{R_f}~\leftrightarrow~\ket{R_i}$ and $\ket{\Delta}~\leftrightarrow~-\ket{\Delta}$. Since this corresponds to the Hermitian-conjugate hopping term, the phase should change its sign, $\theta \rightarrow -\theta$. Unfortunately, there are terms that are even on this "reverse" operation, the $\uwave{curly-underlined}$ terms. However, the requirements,  [$^{\forall} \vb{R} \in S_{hop},~ T^T A^T \vb{R} \in \mathbb{Z}^2$] or [$T^T A T \in \mathbb{Z}_{2\times 2}$], make these $\uwave{curly-underlined}$ terms $0~or~\pi~(mod~2\pi)$, making them effectively odd on the "reverse" operation. Hence, our requirements allow our TB Hamiltonian to be consistent for the reverse hoppings.

Now, let's look at two terms of $\Lambda(\vb{r})$. In fact, the first term in $\Lambda(\vb{r})$, $-\frac{h}{q} ((V^T)^{-1} A T \vb{i}(\vb{r}))\cdot \vtremain*{\vb{r}}
$, is enough to get a vector potential that minimizes the size of the TB Hamiltonian. Then, what is the role of the second term, $-\frac{h}{2q}(\vtremain*{\vb{r}})^T (V^T)^{-1} A V^{-1}\vtremain*{\vb{r}}$? Take a look at $\braket{}_{sym}$ and $\braket{}_{asym}$ terms of the phases shown above. Since $A_{asym}=\frac{\Phi}{2}\big(\begin{smallmatrix} &-1\\1& \end{smallmatrix}\big)$, $\braket{}_{asym}$ terms are not dependent on how we choose a specific $A$, while $\braket{}_{sym}$ terms are dependent on the choice of $A$. The second term in $\Lambda(\vb{r})$, is added to erase the $A$-dependent last term in $\theta_{lin}^{\cdot}$, $\pi(\braket{R_f|R_f}_{sym}-\braket{R_i|R_i}_{sym})$. You can find that this term is vanished in $\theta_{non}^{\cdot}$. Hence, $\uwave{\pi \braket{\Delta | \Delta}_{sym}}$ became the only $A$-dependent term in $\theta_{non}^{\cdot}$.
\end{proof}

\begin{lemma} \label{lemma:MMT}
    For arbitrary $2\times 2$ matrix $M$, $M \mqty(&-1\\1&) M^T = |M| \mqty(&-1\\1&)$.
\end{lemma}
\begin{proof}
    Let $M=\mqty(a&b\\c&d)$.
    \begin{equation}\nonumber
    \begin{split}
        M \mqty(&-1\\1&) M^T = \mqty(a&b\\c&d) \mqty(&-1\\1&) \mqty(a&c\\b&d)
        = \mqty(a&b\\c&d) \mqty(-b&-d\\a&c)
        =|M| \mqty(&-1\\1&)
    \end{split}
    \end{equation}
\end{proof}

\begin{lemma}\label{lemma:cx+dy}
    When $a,b,c,d \in \mathbb{Z}$ satisfy $(a,b)=(c,d)=1$,
    $(ad-bc, c \inv_b a + d \inv_a b) = 1$.

\end{lemma}
\begin{proof}
    \begin{equation}\nonumber
        \mqty(a&b\\c&d) \mqty(\inv_b a & \inv_d c\\ \inv_a b & \inv_c d) = \mqty(1 & a\inv_d c+b\inv_c d \\ c\inv_b a+d\inv_a b & 1)
    \end{equation}

    Taking the determinant of this equation,
    \begin{equation}\nonumber
    \begin{split}
        (ad-bc)(\inv_b a \inv_c d - \inv_a b \inv_d c)+(a\inv_d c+b\inv_c d)(c\inv_b a+d\inv_a b) = 1.
    \end{split}
    \end{equation}

    Hence, $(ad-bc, c \inv_b a + d \inv_a b) = 1$.
\end{proof}

\begin{lemma} \label{eq:ax+by=cZ}
    When $a,b,c \in \mathbb{Z}$, the matrix $T \in \mathbb{Z}_{2\times2}$ such that
    \begin{equation}\nonumber
        T \times \mathbb{Z}^2 = \{\mqty(x\\y) \in\mathbb{Z}^2 | ax+by \in c \mathbb{Z}\}
    \end{equation}
    can be selected as follows:

    (1) If $a=b=0,~ T=1.$
    
\bigskip
    (2) If $a$ or $b$ is nonzero $~i.e.~ (a,b) \neq 0$,
    \begin{equation}\nonumber
        T = \mqty(\inv_b a & -b \\ \inv_a b & a)
        \mqty( \dfrac{c}{(a,b,c)} & \\ & \dfrac{1}{(a,b)} )
        ~and~ \mqty|T| = \dfrac{c}{(a,b,c)}.
    \end{equation}

\bigskip
    \indent\indent (2-1) If $b \neq 0$, the matrix $T$ can be further simplified as
    \begin{equation}\nonumber
        T = \mqty( & -b \\ 1 & a) \mqty( (a,b) & \\ -\inv_b a & 1)
        \mqty(\dfrac{c}{b(a,b,c)} & \\ & \dfrac{1}{(a,b)}).
    \end{equation}

    Any solution $T$ can be derived by applying a transform $T\rightarrow TM$ to the aforementioned solution for some matrix $M\in\mathbb{Z}_{2\times 2}$ such that $\det M =1$.
\end{lemma}

\begin{proof}

(1) If $a=b=0$, $T=1$.

\bigskip
(2) If $a$ or $b$ is nonzero, $(a,b)$ and $(a,b,c)$ are nonzero.

Given $x,y\in \mathbb{Z}$, assume that $^{\exists} t \in \mathbb{Z} ~s.t.~ ax+by=ct$.

Since $(a,b)|ct$, $^{\exists} k \in \mathbb{Z} ~s.t.~ t= \frac{(a,b)}{(a,b,c)} k$.

$ax+by=\dfrac{(a,b)c}{(a,b,c)} k=\dfrac{ck}{(a,b,c)} (a\inv_b a + b\inv_a b)$

$ \Rightarrow a(x-\dfrac{ck}{(a,b,c)} \inv_b a)+b(y-\dfrac{ck}{(a,b,c)} \inv_a b)=0$

$\Rightarrow \mqty(x\\y)=\mqty(\dfrac{c}{(a,b,c)} \inv_b a & -\dfrac{b}{(a,b)}\\ \dfrac{c}{(a,b,c)} \inv_a b & \dfrac{a}{(a,b)}) \times \mqty(k\\l)~for~some~l \in \mathbb{Z}$

When we define $T \defeq \mqty(\inv_b a & -b \\ \inv_a b & a)
        \mqty( \frac{c}{(a,b,c)} & \\ & \frac{1}{(a,b)} )$, 
since $\mqty(a & b) T = \dfrac{(a,b)c}{(a,b,c)}\mqty(1 & 0\\ 0& 0)$,
$^{\forall} \mqty(x\\y)\in T\times \mathbb{Z}^2,~ ax+by \in c \mathbb{Z}$.

$|T| = \mqty|\inv_b a & -b\\ \inv_a b & a|\mqty|\dfrac{c}{(a,b,c)} &  \\  & \dfrac{1}{(a,b)}|=\dfrac{c}{(a,b,c)}$.

\bigskip

\indent\indent (2-1) If $b \neq 0$, since $\inv_a b = \frac{(a,b)-a \inv_b a}{b}$,
\begin{equation}\nonumber
    \begin{split}
        T & = \mqty(\inv_b a & -b \\ -\dfrac{a}{b} \inv_b a + \dfrac{(a,b)}{b} & a) \mqty(\dfrac{c}{(a,b,c)} & \\ & \dfrac{1}{(a,b)})
        \\ & = \mqty(0 & -b \\ 1 & a) \mqty(\dfrac{(a,b)}{b} &\\-\dfrac{\inv_b a}{b} & 1) \mqty(\dfrac{c}{(a,b,c)} & \\ & \dfrac{1}{(a,b)})
        \\ & = \mqty(0 & -b \\ 1 & a) \mqty((a,b) &\\-\inv_b a & 1) \mqty(\dfrac{c}{b(a,b,c)} & \\ & \dfrac{1}{(a,b)}).
    \end{split}
\end{equation}
\end{proof}

\begin{theorem}\label{theorem:|T|=q^2/(q,ad-bc)}
    Given a matrix $Q\in\mathbb{Q}_{2\times 2}$,
    
    The matrix $T \in \mathbb{Z}_{2\times2}$ such that
    \begin{equation}\nonumber
        T \times \mathbb{Z}^2 = \{\mqty(x\\y) \in\mathbb{Z}^2 | Q \mqty(x\\y) \in \mathbb{Z}^2\}
    \end{equation}
    satisfies $|T|=\frac{q_Q^2}{(q_Q,|\Tilde{Q}|)}$ and a solution $T$ can be selected as follows:

    (1) If $Q = 0 ~ i.e. ~ rank(Q) = 0$, $T=1$.

\bigskip
    
    (2) If $Q \neq 0$ and $|Q| = 0 ~ i.e. ~ rank(Q) = 1$,
    \begin{equation}\nonumber
        T = \mqty(\inv_{\Tilde{Q}_{\ast2}} \Tilde{Q}_{\ast1} & -\Tilde{Q}_{\ast2} \\ \inv_{\Tilde{Q}_{\ast1}} \Tilde{Q}_{\ast2} & \Tilde{Q}_{\ast1}) \mqty(q_Q&\\&1).
    \end{equation}

\bigskip

    (3) If $|Q|\neq 0 ~ i.e. ~ rank(Q) = 2$,
    \begin{equation}\nonumber
        T = q_Q |\Tilde{Q}| \Tilde{Q}^{-1}
        \mqty(\dfrac{1}{\Tilde{Q}_{2\ast}(q_Q,\dfrac{|\Tilde{Q}|}{\Tilde{Q}_{2\ast}})} & 
        \dfrac{q_Q\Tilde{Q}_{1\ast}}{|\Tilde{Q}|(q_Q,|\Tilde{Q}|)}\inv_{\frac{|\Tilde{Q}|}{\Tilde{Q}_{1\ast}}} \left(\dfrac{q_Q}{(\Tilde{Q}_{1\ast},q_Q)} (\Tilde{Q}_{21}\inv_{\Tilde{Q}_{12}} \Tilde{Q}_{11} + \Tilde{Q}_{22} \inv_{\Tilde{Q}_{11}} \Tilde{Q}_{12})\right) \\ &
        \dfrac{\Tilde{Q}_{2\ast}(q_Q,\dfrac{|\Tilde{Q}|}{\Tilde{Q}_{2\ast}})}{|\Tilde{Q}|(q_Q,|\Tilde{Q}|)}
        ).
    \end{equation}

    Any solution $T$ can be derived by applying a transform $T\rightarrow TM$ to the aforementioned solution for some matrix $M\in\mathbb{Z}_{2\times 2}$ such that $\det M =1$.
\end{theorem}
\begin{proof}
    (1) If $Q=0$, it's trivial that $T=1$.

Setting $q_Q=1$ from $Q=\dfrac{0}{1}\mqty(0&0\\0&0)$, $|T|=1=\frac{1^2}{(1,0)}$.

\bigskip

    (2) If $Q \neq 0$ and $|Q| = 0$,

    Because $|Q|=0$,
    $Q=\dfrac{p_Q}{q_Q} \mqty(\Tilde{Q}_{1\ast}\\ \Tilde{Q}_{2\ast}) \mqty(\Tilde{Q}_{\ast 1} & \Tilde{Q}_{\ast 2})$.
    
    Since $(\Tilde{Q}_{1\ast},\Tilde{Q}_{2\ast})=(p_Q,q_Q)=1$, $Q \mqty(x\\y) \in \mathbb{Z} \Leftrightarrow \Tilde{Q}_{\ast 1}x+\Tilde{Q}_{\ast 2}y \in q_Q \mathbb{Z}$

    Using case (2) of the Lemma \ref{eq:ax+by=cZ}, because $(\Tilde{Q}_{\ast 1},\Tilde{Q}_{\ast 2})=1$, a solution $T$ is
    \begin{equation}\nonumber
        T=\mqty(\inv_{\Tilde{Q}_{\ast 2}} \Tilde{Q}_{\ast 1} & -\Tilde{Q}_{\ast 2} \\ \inv_{\Tilde{Q}_{\ast 1}} \Tilde{Q}_{\ast 2} & \Tilde{Q}_{\ast 1}) \mqty(q_Q&\\&1).
    \end{equation}

    Also note that $|T|=q_Q=\frac{q_Q^2}{(q_Q,0)}$.

\bigskip

    (3) If $|Q| \neq 0$,
    
    Since $(p_Q,q_Q)=1$, 
    $Q \mqty(x\\y) \in \mathbb{Z}^2
    \Leftrightarrow \Tilde{Q} \mqty(x\\y)\in q_Q \mathbb{Z}^2$.

    Because $\Tilde{Q}_{11}=\Tilde{Q}_{12}=0 \rightarrow |Q|=0$ leads to a contradiction, $\Tilde{Q}_{11}$ or $\Tilde{Q}_{12}$ is nonzero.

    Hence from $\Tilde{Q}_{11}x+\Tilde{Q}_{12}y \in q_Q\mathbb{Z}$, according to the case (2) of Lemma \ref{eq:ax+by=cZ},

    \begin{equation}\nonumber
        \mqty(x\\y)=\mqty(\inv_{\Tilde{Q}_{12}} \Tilde{Q}_{11}&-\Tilde{Q}_{12}\\\inv_{\Tilde{Q}_{11}} \Tilde{Q}_{12}&\Tilde{Q}_{11}) \mqty(\dfrac{q_Q}{(\Tilde{Q}_{1\ast},q_Q)}&\\&\dfrac{1}{\Tilde{Q}_{1\ast}}) \mqty(k\\l)
        ~for~some~k,~l \in \mathbb{Z}
    \end{equation}

    Defining $a,~b\in\mathbb{Z}$ as
    \begin{equation}\nonumber
    a\defeq \dfrac{q_Q}{(\Tilde{Q}_{1\ast},q_Q)} (\Tilde{Q}_{21}\inv_{\Tilde{Q}_{12}} \Tilde{Q}_{11}+\Tilde{Q}_{22}\inv_{\Tilde{Q}_{11}} \Tilde{Q}_{12}),~ b\defeq \dfrac{|\Tilde{Q}|}{\Tilde{Q}_{1\ast}}\neq 0,
    \end{equation}

    \begin{equation}\nonumber
    \begin{split}
        ak+bl& = \mqty(a&b) \mqty(k\\l)
        \\&=\mqty(\Tilde{Q}_{21}&\Tilde{Q}_{22}) \mqty(\inv_{\Tilde{Q}_{12}} \Tilde{Q}_{11}&-\Tilde{Q}_{12}\\\inv_{\Tilde{Q}_{11}} \Tilde{Q}_{12}&\Tilde{Q}_{11}) \mqty(\dfrac{q_Q}{(\Tilde{Q}_{1\ast},q)}&\\&\dfrac{1}{\Tilde{Q}_{1\ast}}) \mqty(k\\l)
        \\&=\mqty(\Tilde{Q}_{21}&\Tilde{Q}_{22}) \mqty(x\\y) = \Tilde{Q}_{21}x+\Tilde{Q}_{22}y \in q_Q \mathbb{Z}
    \end{split}
    \end{equation}

    According to the case (2-1) of Lemma \ref{eq:ax+by=cZ},
    \begin{equation}\nonumber
        \mqty(k\\l) \in \mqty(&-b\\1&a) \mqty((a,b)&\\-\inv_b a&1)
        \mqty(\dfrac{q_Q}{B(a,b,q_Q)}&\\&\dfrac{1}{(a,b)}) \times \mathbb{Z}^2.
    \end{equation}

    So
    \begin{equation}\nonumber
        \mqty(x\\y) \in \mqty(\inv_{\Tilde{Q}_{12}} \Tilde{Q}_{11}&-\Tilde{Q}_{12}\\\inv_{\Tilde{Q}_{11}} \Tilde{Q}_{12}&\Tilde{Q}_{11}) \mqty(\dfrac{q_Q}{(\Tilde{Q}_{1\ast},q)}&\\&\dfrac{1}{\Tilde{Q}_{1\ast}}) \mqty(&-b\\1&a) \mqty((a,b)&\\-\inv_b a&1)
        \mqty(\dfrac{q_Q}{B(a,b,q_Q)}&\\&\dfrac{1}{(a,b)}) \times \mathbb{Z}^2
    \end{equation}

    According to the Lemma \ref{lemma:MMT},
    \begin{equation}\nonumber
    \begin{split}
        &\mqty(\inv_{\Tilde{Q}_{12}} \Tilde{Q}_{11}&-\Tilde{Q}_{12}\\\inv_{\Tilde{Q}_{11}} \Tilde{Q}_{12}&\Tilde{Q}_{11}) \mqty(\dfrac{q_Q}{(\Tilde{Q}_{1\ast},q)}&\\&\dfrac{1}{\Tilde{Q}_{1\ast}}) \mqty(-b\\a) \\=& 
        \mqty(\inv_{\Tilde{Q}_{12}} \Tilde{Q}_{11}&-\Tilde{Q}_{12}\\\inv_{\Tilde{Q}_{11}} \Tilde{Q}_{12}&\Tilde{Q}_{11}) \mqty(\dfrac{q_Q}{(\Tilde{Q}_{1\ast},q)}&\\&\dfrac{1}{\Tilde{Q}_{1\ast}}) \mqty(&-1\\1&) \mqty(\dfrac{q_Q}{(\Tilde{Q}_{1\ast},q)}&\\&\dfrac{1}{\Tilde{Q}_{1\ast}}) \mqty(\inv_{\Tilde{Q}_{12}} \Tilde{Q}_{11}&\inv_{\Tilde{Q}_{11}} \Tilde{Q}_{12}\\-\Tilde{Q}_{12}&\Tilde{Q}_{11}) \mqty(\Tilde{Q}_{21}\\\Tilde{Q}_{22})
        \\=&\mqty|\inv_{\Tilde{Q}_{12}} \Tilde{Q}_{11}&-\Tilde{Q}_{12}\\\inv_{\Tilde{Q}_{11}} \Tilde{Q}_{12}&\Tilde{Q}_{11}| \mqty|\dfrac{q_Q}{(\Tilde{Q}_{1\ast},q)}&\\&\dfrac{1}{\Tilde{Q}_{1\ast}}| \mqty(&-1\\1&) \mqty(\Tilde{Q}_{21}\\\Tilde{Q}_{22})
        \\=&\dfrac{q_Q}{(\Tilde{Q}_{1\ast},q_Q)} \mqty(-\Tilde{Q}_{22}\\\Tilde{Q}_{21}).
    \end{split}\end{equation}

    So,
    \begin{equation}\nonumber
    \begin{split}
        \mqty(\inv_{\Tilde{Q}_{12}} \Tilde{Q}_{11}&-\Tilde{Q}_{12}\\\inv_{\Tilde{Q}_{11}} \Tilde{Q}_{12}&\Tilde{Q}_{11}) \mqty(\dfrac{q_Q}{(\Tilde{Q}_{1\ast},q)}&\\&\dfrac{1}{\Tilde{Q}_{1\ast}}) \mqty(&-b\\1&a) = 
        &\mqty(-\dfrac{\Tilde{Q}_{12}}{\Tilde{Q}_{1\ast}} & -\dfrac{q_Q \Tilde{Q}_{22}}{(\Tilde{Q}_{1\ast},q_Q)} \\ \dfrac{\Tilde{Q}_{11}}{\Tilde{Q}_{1\ast}} & \dfrac{q_Q \Tilde{Q}_{21}}{(\Tilde{Q}_{1\ast},q_Q)})
        \\&=\mqty(-\Tilde{Q}_{12}&-\Tilde{Q}_{22}\\\Tilde{Q}_{11}&\Tilde{Q}_{21}) \mqty(\dfrac{1}{\Tilde{Q}_{1\ast}}&\\&\dfrac{q_Q}{(\Tilde{Q}_{1\ast},q_Q)})
    \end{split}        
    \end{equation}

    Then,
    \begin{equation}\begin{split}\label{eq:before}
        \mqty(x\\y) \in &\mqty(-\Tilde{Q}_{12}&-\Tilde{Q}_{22}\\\Tilde{Q}_{11}&\Tilde{Q}_{21}) \mqty(\dfrac{1}{\Tilde{Q}_{1\ast}}&\\&\dfrac{q_Q}{(\Tilde{Q}_{1\ast},q_Q)}) \mqty((a,b)&\\-\inv_b a&1)
        \mqty(\dfrac{q_Q}{b(a,b,q_Q)}&\\&\dfrac{1}{(a,b)}) \times \mathbb{Z}^2
        \\=&\mqty(-\Tilde{Q}_{12}&-\Tilde{Q}_{22}\\\Tilde{Q}_{11}&\Tilde{Q}_{21}) \mqty(\dfrac{q_Q(a,b)}{b\Tilde{Q}_{1\ast}(a,b,q_Q)} & \\
        -\dfrac{{q_Q}^2 \inv_b a}{b(\Tilde{Q}_{1\ast},q_Q)(a,b,q_Q)} & \dfrac{q_Q}{(\Tilde{Q}_{1\ast},q_Q)(a,b)}) \times \mathbb{Z}^2
        \\=&\mqty(\Tilde{Q}_{22}&-\Tilde{Q}_{12}\\-\Tilde{Q}_{21}&\Tilde{Q}_{11}) \mqty(\dfrac{{q_Q}^2 \inv_b a}{b(\Tilde{Q}_{1\ast},q_Q)(a,b,q_Q)} & -\dfrac{q_Q}{(\Tilde{Q}_{1\ast},q_Q)(a,b)} \\ \dfrac{q_Q(a,b)}{b\Tilde{Q}_{1\ast}(a,b,q_Q)}&) \times \mathbb{Z}^2
        \\=&|\Tilde{Q}| \Tilde{Q}^{-1} \mqty(\dfrac{q_Q}{(\Tilde{Q}_{1\ast},q_Q)(a,b)} & \dfrac{{q_Q}^2 \inv_b a}{b(\Tilde{Q}_{1\ast},q_Q)(a,b,q_Q)} \\ &\dfrac{q_Q(a,b)}{b\Tilde{Q}_{1\ast}(a,b,q_Q)}) \times \mathbb{Z}^2
    \end{split}\end{equation}

    \begin{equation}\begin{split}\nonumber
        (a,b)=&\Tilde{Q}_{2\ast} (\dfrac{q_Q}{(\Tilde{Q}_{1\ast},q_Q)}\dfrac{\Tilde{Q}_{21}\inv_{\Tilde{Q}_{12}} \Tilde{Q}_{11}+\Tilde{Q}_{22}\inv_{\Tilde{Q}_{11}} \Tilde{Q}_{12}}{\Tilde{Q}_{2\ast}},~\dfrac{|\Tilde{Q}|}{\Tilde{Q}_{1\ast}\Tilde{Q}_{2\ast}})
        \\=&\Tilde{Q}_{2\ast}(\dfrac{q_Q}{(\Tilde{Q}_{1\ast},q_Q)},~\dfrac{|\Tilde{Q}|}{\Tilde{Q}_{1\ast}\Tilde{Q}_{2\ast}}) (\because Lemma~\ref{lemma:cx+dy})
        \\=&\dfrac{\Tilde{Q}_{2\ast}}{(\Tilde{Q}_{1\ast},q_Q)} (q_Q,~\dfrac{|\Tilde{Q}|}{\Tilde{Q}_{1\ast}\Tilde{Q}_{2\ast}}(\Tilde{Q}_{1\ast},q_Q))
        =\dfrac{\Tilde{Q}_{2\ast}}{(\Tilde{Q}_{1\ast},q_Q)} (q_Q,~\dfrac{|\Tilde{Q}|}{\Tilde{Q}_{1\ast}\Tilde{Q}_{2\ast}}\Tilde{Q}_{1\ast},~\dfrac{|\Tilde{Q}|}{\Tilde{Q}_{1\ast}\Tilde{Q}_{2\ast}}q_Q)
        \\=&\dfrac{\Tilde{Q}_{2\ast}}{(\Tilde{Q}_{1\ast},q_Q)} (q_Q,~\dfrac{|\Tilde{Q}|}{\Tilde{Q}_{2\ast}})
    \end{split}\end{equation}

    \begin{equation}\begin{split}\nonumber
        (a,b,q_Q)=&(\dfrac{\Tilde{Q}_{2\ast}}{(\Tilde{Q}_{1\ast},q_Q)} (q_Q,~\dfrac{|\Tilde{Q}|}{\Tilde{Q}_{2\ast}}),~q_Q)
        =\dfrac{1}{(\Tilde{Q}_{1\ast},q_Q)}(q_Q\Tilde{Q}_{2\ast},~|\Tilde{Q}|,~q_Q(\Tilde{Q}_{1\ast},q_Q))
        \\=&\dfrac{1}{(\Tilde{Q}_{1\ast},q_Q)}(q_Q(\Tilde{Q}_{1\ast},\Tilde{Q}_{2\ast},q_Q),~|\Tilde{Q}|)
        =\dfrac{(q_Q,|\Tilde{Q}|)}{(\Tilde{Q}_{1\ast},q_Q)}~(\because (\Tilde{Q}_{1\ast},\Tilde{Q}_{2\ast})=1)
    \end{split}\end{equation}

    Substituting $a,~b,~(a,~b),~(a,~b,~q_Q)$ into the Equation \ref{eq:before},
    $\mqty(x\\y)\in T \times \mathbb{Z}^2$ where
    \begin{equation}\nonumber
        T = q_Q |\Tilde{Q}| \Tilde{Q}^{-1}
        \mqty(\dfrac{1}{\Tilde{Q}_{2\ast}(q_Q,\dfrac{|\Tilde{Q}|}{\Tilde{Q}_{2\ast}})} & 
        \dfrac{q_Q\Tilde{Q}_{1\ast}}{|\Tilde{Q}|(q_Q,|\Tilde{Q}|)}\inv_{\frac{|\Tilde{Q}|}{\Tilde{Q}_{1\ast}}} \left(\dfrac{q_Q}{(\Tilde{Q}_{1\ast},q_Q)} (\Tilde{Q}_{21}\inv_{\Tilde{Q}_{12}} \Tilde{Q}_{11} + \Tilde{Q}_{22} \inv_{\Tilde{Q}_{11}} \Tilde{Q}_{12})\right) \\ &
        \dfrac{\Tilde{Q}_{2\ast}(q_Q,\dfrac{|\Tilde{Q}|}{\Tilde{Q}_{2\ast}})}{|\Tilde{Q}|(q_Q,|\Tilde{Q}|)}
        )
    \end{equation}
    The determinant of the matrix $T$ is calculated as
    \begin{equation}\nonumber
        |T|=\dfrac{q_Q^2}{(q_Q,|\Tilde{Q}|)}
    \end{equation}.
\end{proof}

\begin{lemma}\label{lemma:(a,cx-b)}
    Given $a,b,c \in \mathbb{Z}$,
    
    (1) When (a,c)=1,
    \begin{equation}\nonumber
        \{-b,~c-b,~2c-b,~...,~c(a-1)-b\} \equiv \{0,~1,~2,...,~a-1\}~(mod~ a).
    \end{equation}

\bigskip

    (2) When $(b,c)=1$, the solution $x\in\mathbb{Z}$ such that $a$ and $cx-b$ are coprime (i.e. $(a,cx-b)=1$) is $1+b\inv_{\erase_c a} c$.
    
    Ambiguity of the definition of $\inv$ function explains all solutions.
\end{lemma}
\begin{proof}
    (1) For two integers $k,l$ such that $0\leq k,l < a$,
    \begin{equation}\nonumber
        ck-b \equiv cl-b ~(mod~a) \Rightarrow c(k-l) \equiv 0 ~(mod~a) \Rightarrow k-l \equiv 0 ~(mod~a) \Rightarrow k=l.
    \end{equation}
    As a result, the remainders of $a$ integers $-b,~c-b,~2c-b,~...,~c(a-1)-b$ with respect to the modulus $a$ are all different, hence
    \begin{equation}\nonumber
        \{-b,~c-b,~2c-b,~...,~c(a-1)-b\} \equiv \{0,~1,~2,...,~a-1\}~(mod~ a).
    \end{equation}

\bigskip

    (2) Since $(c,cx-b)=(c,-b)=1$, $(a,cx-b)=1\Leftrightarrow (\erase_c a, cx-b)=1$.
    
    Since $(\erase_c a,c)=1$, the proof of (1) shows that the solution exists and the difference between any two solutions is a multiple of $\erase_c a$. Hence, it is enough to show that $1+b\inv_{\erase_c a} c$ is a solution.

Because

    \begin{equation}\nonumber
    \erase_c a \inv_c (\erase_c a) + c \inv_{\erase_c a} c =1
    ~\Rightarrow~
    \erase_c a~|~1-c \inv_{\erase_c a} c
    \end{equation}

and

    \begin{equation}\nonumber
    (1-c \inv_{\erase_c a} c,~c(1+b\inv_{\erase_c a} c)-b)=(1-c \inv_{\erase_c a} c,c)=(1,c)=1,
    \end{equation}

    $\erase_c a$ and $c(1+b\inv_{\erase_c a} c)-b$ are coprime.
    
\end{proof}

\newpage

\begin{lemma}\label{lemma:(ax,by)=(ax,y)(x,by)/(x,y)}
    For $a,b,x,y\in\mathbb{Z}$ satisfying $(a,b)=1$,
    \begin{equation}\nonumber
        (ax,by)=\dfrac{(ax,y)(x,by)}{(x,y)}.
    \end{equation}
\end{lemma}
\begin{proof}
    \begin{equation}\begin{split}\nonumber
        (ax,by)= & (x,y)(a\dfrac{x}{(x,y)},b\dfrac{y}{(x,y)})
        \\=& (x,y)(a,\dfrac{y}{(x,y)})(\dfrac{x}{(x,y)},b)~(\because
        (a,b)=(\dfrac{x}{(x,y)},\dfrac{y}{(x,y)})=1)
        \\=& (x,y)\dfrac{(a(x,y),y)}{(x,y)}\dfrac{(x,b(x,y))}{(x,y)}
        = \dfrac{(ax,ay,y)(x,bx,by)}{(x,y)}
        = \dfrac{(ax,y)(x,by)}{(x,y)}.
    \end{split}\end{equation}
\end{proof}

\begin{lemma}\label{lemma:(a,b,c)=(x,y)=1=>(ax+by,cy)|ac and max is ac}
    For $a,b,c\in\mathbb{Z}$ satisfying $(a,b,c)=1$,

    (1)
    $
        ^{\forall} x,~y \in \mathbb{Z}~ s.t.~ (x,y)=1,
        ~(ax+by,cy)|ac
    $

\bigskip

    (2) Solution $(x,y)\in\mathbb{Z}^2$ of 
    $\begin{cases} (x,y)=1 \\ (ax+by,cy)=ac \end{cases}$ is 
    
    \indent\indent $(c+\frac{bck}{(b,c)}\inv_{\erase_{\left(\frac{c}{(b,c)}\right)} ak} \left(\frac{c}{(b,c)}\right)-bk,~ak)$
    
    \indent\indent where $k$ can be any integer satisfying $(k,b,c)=1$.

    \indent\indent$k$ and the ambiguity of the definition of $\inv$ function explains all solutions.
\end{lemma}
\begin{proof}
(1) Given $x,y$ such that $(x,y)=1$,
    
    define integers $t_1 \equiv (\frac{a}{(a,b)},y),~
    t_2 \equiv \frac{(a,y)}{t_1},~
    t_3 \equiv \frac{y}{(a,y)},~
    t_4 \equiv \frac{\frac{a}{(a,b)}}{t_1},~
    t_5 \equiv (a,b),~
    t_6 \equiv \frac{b}{(a,b)}.$

    Then, $a=t_1 t_4 t_5,~ b=t_5 t_6,~ y=t_1 t_2 t_3$ and $(t_2 t_3, t_4)=(t_1 t_4, t_6)=(x,t_1 t_2 t_3)=(c,t_5)=1.$

    Also, $t_2 | t_4 t_5$ and $(\frac{t_4 t_5}{t_2},t_3)=(\frac{a}{(a,y)},\frac{y}{(a,y)})=1 \Rightarrow (t_4 t_5,t_2 t_3)=t_2.$
    
    Then,
    \begin{equation}\begin{split}\nonumber
        (ax+by,cy)=&(xt_1 t_4 t_5 + t_1 t_2 t_3 t_5 t_6, c t_1 t_2 t_3)
        = t_1 (t_5(xt_4 + t_2 t_3t_6),c t_2 t_3)
        \\=& t_1 \dfrac{(t_5(xt_4 + t_2 t_3t_6),t_2 t_3)(xt_4 + t_2 t_3t_6,c t_2 t_3)}{(xt_4 + t_2 t_3 t_6,t_2 t_3)}~(\because Lemma~ \ref{lemma:(ax,by)=(ax,y)(x,by)/(x,y)},~(c,t_5)=1)
        \\= & t_1 \dfrac{(xt_4 t_5,t_2 t_3)(xt_4 + t_2 t_3t_6,c t_2 t_3)}{(x t_4, t_2 t_3)}
        \\=& t_1 (t_4 t_5,t_2 t_3)(xt_4 + t_2 t_3t_6,c t_2 t_3)
        (\because (x,t_2 t_3)=(t_4,t_2 t_3)=1)
        \\=&t_1 t_2 (xt_4 + t_2 t_3t_6,c t_2 t_3)
        \\=& t_1 t_2 (xt_4 + t_2 t_3 t_6, c)~(\because(xt_4 + t_2 t_3t_6,t_2 t_3)=(xt_4,t_2 t_3)=1)
        \\=& (a,y) (xt_4 + t_2 t_3 t_6, c)
        \\ \Rightarrow  (ax+by,cy) | & ac.
    \end{split}\end{equation}

\bigskip

(2) If $(x,y)\in\mathbb{Z}^2$ satisfies $(ax+by,cy)=ac$, $x=cl-bk$ and $y=ak$ for some $k,l\in\mathbb{Z}$.

Since $1=(x,y)=(cl-bk,ak)$ should also be satisfied, $(k,b,c)=1$.

Then, because $(\frac{cl-bk}{(b,c)},ak)=1$, according the Lemma \ref{lemma:(a,cx-b)} (2),
$l=1+\dfrac{bk}{(b,c)}\inv_{\erase_{\left(\frac{c}{(b,c)}\right)} ak} \dfrac{c}{(b,c)}$ is a solution.
\end{proof}

\begin{theorem}\label{theorem:main theorem for linear vector potential}
Given $D=\gcd D \mqty(\Tilde{D}_{11} & \\ \Tilde{D}_{21} & \Tilde{D}_{22})\in \mathbb{Q}_{2\times 2}$ and $\Phi \in \mathbb{Q}$,

let's think about solution $A\in\mathbb{Q}_{2\times 2},~T\in\mathbb{Z}_{2\times 2}$ satisfying
$\begin{cases}
A_{21}-A_{12}=\Phi \\ 
|T|>0 \\
D^T A T \in \mathbb{Z}_{2\times 2}
\end{cases}$.

(1) $q_{\Phi|\Tilde{D}| D }~|~|T|$

\bigskip

(2) Define an integer
$\delta \defeq \Tilde{D}_{22}(1+\dfrac{\Tilde{D}_{21}}{(\Tilde{D}_{21},\Tilde{D}_{22})} \inv_{\left(\erase_{\frac{\Tilde{D}_{22}}{(\Tilde{D}_{21},\Tilde{D}_{22})}} \Tilde{D}_{11}\right)} \dfrac{\Tilde{D}_{22}}{(\Tilde{D}_{21},\Tilde{D}_{22})})-\Tilde{D}_{21}$.

Then,
\begin{equation}\begin{split}\nonumber
A &= \Phi \mqty(\delta \\ \Tilde{D}_{11}) \mqty(\inv_{\delta} \Tilde{D}_{11} & -\inv_{\Tilde{D}_{11}} \delta),
\\
T &=\mqty(\Tilde{D}_{11} & \inv_{\Tilde{D}_{11}}{\delta}
    \\ -\delta & \inv_{\delta}{\Tilde{D}_{11}})
    \mqty(q_{\Phi|\Tilde{D}| D} & 
    \\ & 1)
\end{split}\end{equation}
is a solutions that has the smallest $|T|=q_{\Phi|\Tilde{D}| D }$.
\end{theorem}
\begin{proof}
(1)
Assume that $A$ and $T$ is a solution.

According to the Theorem \ref{theorem:|T|=q^2/(q,ad-bc)}, it is enough to prove that $q_{\Phi|\Tilde{D}| D }~|~\frac{q_Q^2}{(q_Q ,|\Tilde{Q}|)}$ where $Q$ is defined as $D^T A$.

First, we will show that $\gcd (\Tilde{D}^{T}\Tilde{A})~|~|\Tilde{D}|$.

\begin{equation}\nonumber
\Tilde{D}^{T}\Tilde{A}=\mqty(\Tilde{D}_{11} \Tilde{A}_{11}+\Tilde{D}_{21} \Tilde{A}_{21} & \Tilde{D}_{11} \Tilde{A}_{12}+\Tilde{D}_{21} \Tilde{A}_{22} \\
\Tilde{D}_{22} \Tilde{A}_{21} & \Tilde{D}_{22} \Tilde{A}_{22}).
\end{equation}
      
    \indent \indent 1) When $\Tilde{A}_{\ast 1},\Tilde{A}_{\ast 2}\neq0$
    
    According to the Lemma \ref{lemma:(a,b,c)=(x,y)=1=>(ax+by,cy)|ac and max is ac} (1),
    since $(\frac{\Tilde{A}_{11}}{\Tilde{A}_{\ast1}},\frac{\Tilde{A}_{21}}{\Tilde{A}_{\ast1}})=(\Tilde{D}_{11},\Tilde{D}_{21},\Tilde{D}_{22})=1$,

    \begin{equation}\begin{split}\nonumber
    & (\Tilde{D}_{11} \dfrac{\Tilde{A}_{11}}{\Tilde{A}_{\ast 1}} + \Tilde{D}_{21} \dfrac{\Tilde{A}_{21}}{\Tilde{A}_{\ast 1}},\Tilde{D}_{22} \dfrac{\Tilde{A}_{21}}{\Tilde{A}_{\ast 1}})~|~\Tilde{D}_{11}\Tilde{D}_{22}=|\Tilde{D}| \\
    \Rightarrow & 
    (\Tilde{D}_{11} \Tilde{A}_{11} + \Tilde{D}_{21} \Tilde{A}_{21},\Tilde{D}_{22} \Tilde{A}_{21})~|~\Tilde{A}_{\ast 1}|\Tilde{D}|
    \end{split}\end{equation}

    In the same way,
    \begin{equation}\begin{split}\nonumber
    (\Tilde{D}_{11} \Tilde{A}_{12} + \Tilde{D}_{21} \Tilde{A}_{22},\Tilde{D}_{22} \Tilde{A}_{22})~|~\Tilde{A}_{\ast 2}|\Tilde{D}|
    \end{split}\end{equation}

    Then,
    \begin{equation}\nonumber
    \gcd (\Tilde{D}^{T}\Tilde{A})~|~ (\Tilde{A}_{\ast 1},\Tilde{A}_{\ast 2})|\Tilde{D}|=|\Tilde{D}|.
    \end{equation}

\bigskip

    \indent \indent 2) When $\Tilde{A}_{\ast 1}$ or $\Tilde{A}_{\ast 2}$ is zero

    Without loss of generality, assume that $\Tilde{A}_{\ast 2}=1,~\Tilde{A}_{\ast 2}=0$.

    In the same way as the case 1),
    \begin{equation}\begin{split}\nonumber
    (\Tilde{D}_{11} \Tilde{A}_{11} + \Tilde{D}_{21} \Tilde{A}_{21},\Tilde{D}_{22} \Tilde{A}_{21})~|~\Tilde{A}_{\ast 1}|\Tilde{D}|.
    \end{split}\end{equation}

    Hence,
    \begin{equation}\nonumber
    \gcd (\Tilde{D}^{T}\Tilde{A})~|~ (\Tilde{A}_{\ast 1},0)|\Tilde{D}|=|\Tilde{D}|.
    \end{equation}

Since $\gcd (\Tilde{D}^{T}\Tilde{A})~|~|\Tilde{D}|$ and $\Phi=\gcd A (\Tilde{A}_{21}-\Tilde{A}_{12})$,
\begin{equation}\nonumber
q_{\Phi|\Tilde{D}| D }~|~q_{\gcd D \gcd A \gcd (\Tilde{D}^{T}\Tilde{A})}=q_Q.
\end{equation}

Because $q_Q~|~\frac{q_Q^2}{(q_Q ,|\Tilde{Q}|)}$,
\begin{equation}\nonumber
q_{\Phi|\Tilde{D}| D }~|~\dfrac{q_Q^2}{(q_Q ,|\Tilde{Q}|)}.
\end{equation}

\bigskip

(2)
According to the Lemma \ref{lemma:(a,b,c)=(x,y)=1=>(ax+by,cy)|ac and max is ac} (2),
$\Tilde{D}_{11} \inv_{\delta} \Tilde{D}_{11} + \delta \inv_{\Tilde{D}_{11}} \delta = (\delta,~\Tilde{D}_{11})=1$.

\begin{equation}\nonumber
A_{21}-A_{12}=\Phi (\Tilde{D}_{11} \inv_{\delta} \Tilde{D}_{11} + \delta \inv_{\Tilde{D}_{11}} \delta)=\Phi.
\end{equation}

\begin{equation}\nonumber
|T|=\mqty|\Tilde{D}_{11} & \inv_{\Tilde{D}_{11}}{\delta}
    \\ -\delta & \inv_{\delta}{\Tilde{D}_{11}}|
    \mqty|q_{\Phi|\Tilde{D}| D} & 
    \\ & 1|
    = q_{\Phi|\Tilde{D}| D}>0.
\end{equation}

\begin{equation}\begin{split}\nonumber
D^T A &=\Phi \gcd D \mqty(\Tilde{D}_{11} & \Tilde{D}_{21} \\  & \Tilde{D}_{22}) \mqty(\delta \\ \Tilde{D}_{11}) \mqty(\inv_{\delta} \Tilde{D}_{11} & -\inv_{\Tilde{D}_{11}} \delta)\\
&=\Phi|\Tilde{D}| \gcd D \mqty(\delta_0 \\ 1)\mqty(\inv_{\delta} \Tilde{D}_{11} & -\inv_{\Tilde{D}_{11}} \delta)
\end{split}\end{equation}
where $\delta_0 = 1+\dfrac{\Tilde{D}_{21}}{(\Tilde{D}_{21},\Tilde{D}_{22})} \inv_{\left(\erase_{\frac{\Tilde{D}_{22}}{(\Tilde{D}_{21},\Tilde{D}_{22})}} \Tilde{D}_{11}\right)} \dfrac{\Tilde{D}_{22}}{(\Tilde{D}_{21},\Tilde{D}_{22})}$.

According to Theorem \ref{theorem:|T|=q^2/(q,ad-bc)} (2), $T$ is a solution to $D^T A T \in \mathbb{Z}_{2\times2}$.
\end{proof}

\end{document}